\documentclass[twoside,a4paper,12pt]{amsart}

\usepackage{amsmath}
\usepackage{amssymb}
\usepackage{amsthm}
\usepackage{amsxtra}
\usepackage{mathrsfs}
\usepackage{graphicx}

\newcommand{\pd}[2]{\frac{\partial {#1}}{\partial {#2}}}
\newcommand{\RE}{\mathbb R}
\newcommand{\CO}{\mathbb C}

\newcommand{\BB}{\mathscr B}
\newcommand{\DD}{\mathscr D}
\newcommand{\HH}{\mathcal{H}}
\newcommand{\OO}{\mathcal{O}}
\newcommand{\QQ}{\mathcal Q}
\newcommand{\ulim}{\operatorname{u}-\lim}

\newcommand{\lf}{\left}
\newcommand{\ri}{\right}
\newcommand{\ve}{\varepsilon}
\newcommand{\al}{\alpha}
\newcommand{\bt}{\beta}
\newcommand{\Ga}{\Gamma}
\newcommand{\ga}{\gamma}
\newcommand{\la}{\lambda}
\newcommand{\de}{\delta}

\newcommand{\De}{\Delta}

\renewcommand{\Im}{\operatorname{Im}\,}

\newcommand{\erre}{\mathbb{R}} 
\newcommand{\n}{\noindent}
\newcommand{\Hd}{\hat H}
\newcommand{\hh}{{h}}
\newcommand{\rbar}{\hat{r}}
\newcommand{\dde}{\de^{\ve\,2}}
\newcommand{\hhd}{{\bf h}}

\newtheorem{theorem}{Theorem}
\newtheorem{lemma}{Lemma}
\newtheorem{proposition}{Proposition}
\theoremstyle{definition}
\theoremstyle{remark}

\author{Claudio Cacciapuoti} 
\address{Cacciapuoti: Czech Technical University, Doppler Institute}
\curraddr{B\v{r}ehov\'a 7, 11000 Prague, Czech Republic}

\email{cacciapuoti@ujf.cas.cz}

\author{Domenico Finco} 
\address{Finco: Department of Mathematics ``G. Castelnuovo", University of Rome ``La Sapienza''}
\curraddr{P.le A.Moro 2, 00185 Rome, Italy }

\email{finco@mat.uniroma1.it}

\title{Graph-like models for thin waveguides with Robin boundary conditions}

\begin{document}

\begin{abstract}
We discuss the limit of small width for the Laplacian defined on a waveguide with Robin boundary conditions. Under suitable hypothesis on the scaling of the curvature, we prove the convergence of the Robin Laplacian to the Laplacian on the corresponding graph. We show that the projections on each transverse mode generically give rise to decoupling conditions between the edges of the graph while exceptionally a coupling can occur. The non decoupling conditions are related to the existence of resonances at the thresholds of the continuum spectrum.

\end{abstract}

\maketitle

\section{Introduction}

The interest in the analysis of  differential operators on metric graphs has been driven by the idea that some physical systems can be well modeled by using lower dimensional approximations. 

Intuitively one expects that graph-like approximations can be used to describe the dynamics in constrained systems characterized by two scales of length: a ``large'' one along the direction of the edges and ``small'' ones in the transverse directions. 

In mathematics a metric graph is a one dimensional singular manifold, and is defined by assigning 
points, the vertices connected by
a set of oriented segments, the edges. The dynamics on the graphs is fixed by defining a differential (or pseudo-differential) operator on the graph and boundary conditions in the vertices. In the following we shall consider only self-adjoint operators. 

When this kind of structures arise as approximations of  quantum systems, it is customary to call them 
quantum graphs. In this case the self-adjointness assumption is natural.
One of the first and most famous applications of quantum graphs dates back to 1953, when they were used to model the dynamics of $\pi$ electrons in organic molecules \cite{RS53}. More recently a renewed interest in quantum graphs has resulted from the  development of nanotechnologies. At present devices based on carbon nanotubes and metallic nanowires are commonly produced and studied. In such structures the mean free path can reach hundreds of micrometers , while the transverse confinement can be of the order of ten nanometers. Because of their small dimensions and purity such devices represent an ideal framework for the analysis of many peculiar phenomena of quantum mechanics, see, e.g., \cite{Hur00} and \cite{LCM99}, and in many applications they can be treated as one dimensional systems. For a comprehensive review on properties and applications of quantum graphs we refer to \cite{Kuc02}, \cite{Kuc04}, \cite{Kuc05} and \cite{BCFK06}. Even in the field of classical mechanics there is a huge number of problems in which metric graphs define simplified but non trivial models. Typical examples arise in the analysis of  electromagnetic or acoustic waves in thin waveguides.

Many efforts have been done in the last fifteen years to understand to which extent a one dimensional dynamics on a metric graph approximates the  dynamics in a constrained system. Aim of this work is to investigate some relevant features of graph-like approximations for the quantum dynamics in networks of thin tubes.

To understand the core of the problem it is sufficient to discuss the case of a graph with one vertex and $N$ infinite edges. The natural Hilbert space for such a system is the direct sum of $N$ copies of $L^2((0,\infty))$ and its generic element is  $(f_1,\dots,f_N)$ with $f_j\in L^2((0,\infty))$. We restrict ourselves to a setting in which the dynamics on the graph is generated by an Hamiltonian that on each edge coincides with  the (positive) Laplacian. From the mathematical point of view one can define several self-adjoint operators on the graph that coincide with the Laplacian on the edges. Each element of this  family of operators is identified by the  boundary conditions in the vertex imposed on the functions in its domain. In the following we shall call the boundary conditions in the vertex \emph{gluing conditions}. A linear relation between $(f_1(0),\dots,f_N(0))$ and $(f_1'(0),\dots,f_N'(0))$ must be used to fix the gluing conditions in a way such that the corresponding operator on the graph is self-adjoint. Making use of Krein's theory one can characterize all the possible self-adjoint gluing conditions in the vertex. This was done by  Kostrykin and Schrader in \cite{KS99} (see also \cite{Har00}); for each vertex of degree $N$ there are $N^2$ real free parameters to fix the gluing conditions.

Two well known  examples of self-adjoint gluing conditions are
\begin{equation}
\label{deccon}
f_j(0)=0\qquad\forall j=1,\dots,N
\end{equation}
and
\begin{equation}
\label{freecon}
f_1(0)=f_2(0)=\dots=f_N(0)\,,\qquad\sum_{j=1}^Nf_j'(0)=0\,. 
\end{equation}
Condition \eqref{deccon} is usually called \emph{decoupling condition} or \emph{Dirichlet condition}. We shall use the expression \emph{decoupling condition} to keep in mind that this type of gluing in the vertex implies decoupling among the edges, i.e., no transmission trough the vertex is possible. Condition \eqref{freecon} is usually referred to as condition of  \emph{free type} or \emph{Kirchhoff type}. In our opinion the  expression condition of Kirchhoff type is a bit misleading. It recalls current conservation in electric circuits but as a matter of fact every self-adjoint gluing condition conserves the quantum probability  current across the vertex. The expression free condition seems more appropriate because \eqref{freecon} generalizes  the free one dimensional Laplacian to a non trivial topology.

In view of applications one is interested in understanding which gluing conditions are more appropriate to model  strongly constrained quantum systems. A standard strategy to approach this problem consists in the analysis of the limit, in some suitable sense, of the Laplacian  in a network of thin tubes as the network shrinks to the underlying graph. 

Let us assume that the network we consider is made up of  tubes that far from the vertex are straight and of constant width. To fix ideas we also suppose that the tubes have all the same width. In this  setting, far from the vertex, the dynamics is factorized in the direction along the axes of the tube and in the transverse direction. Moreover as the manifold shrinks the energy gap between the transverse modes increases as the inverse of the squared width of the tube. Even if this simple picture fails as one approaches to the vertex, it suggests that the natural way to reduce the dynamics  to one on the  underlying  graph is to project onto the transverse modes and that, in the limit of zero width, each projection can lead to a unitary dynamics, that is an effective
dynamics which leaves invariant the subspace associated to the transverse mode. 

The feasibility of this procedure and the corresponding limit operator on the graph depend on the boundary conditions that one imposes on the boundary of the tubes and on which transverse mode the projection is taken.

In this paper we discuss this problem in the most simple geometrical setting. We take as initial domain 
$\Omega$, a strip of constant width $d$ around a base curve $\Gamma$ and we assume that $\Gamma$ has
no self-intersections. Under this assumptions  the model is greatly simplified by the presence of a global system of coordinates
$(s,u)$ adapted to the geometry of the system, where $s$ is the arc length coordinate on $\Gamma$
and $u$ is the orthogonal coordinate. The relevant geometric quantity are $d$ and the scalar curvature
of the curve, $\ga (s)$, that in our hypothesis is a smooth and  compactly supported function. We rescale
the initial domain according to $d \longrightarrow \ve^a d$ and $\ga(s) \longrightarrow \ve^{-1} \ga(s/ \ve)$, where 
$\ve >0$ and $a$ is a big enough positive constant, and we obtain a net of  domains $\Omega^\ve$ which as $\ve\to0$ collapses onto a broken line, i.e., onto a graph made up of one vertex and two infinite edges.
We consider the Laplacian on $\Omega^\ve$ with Robin boundary conditions
and discuss its convergence to the Laplacian on the limit domain which
can be seen as the most simple example of graph.

The most relevant technical part of the present paper consists in the proof that  the projection of the dynamics
on the $n$-th transverse mode is unitary  in the limit $\ve \to 0$, and that the effective dynamics is
described by the Hamiltonian
\begin{equation*}
h_n^\ve = -\frac{d^2}{ds^2} + \frac{\beta_n }{\ve^2} \ga^2 (\cdot / \ve)\,,
\end{equation*}
where $\beta_n$ are some coefficients related to the energy of the transverse mode.

The analysis of the limit for $\ve\to0$ of Hamiltonians of the form of $h_n^\ve$ was performed in the former work by the same authors \cite{ACF07}. For each $n$ the limit operator depends on the low energy properties of the Hamiltonian $h_n$
\begin{equation*}
h_n = -\frac{d^2}{ds^2} + {\beta_n } \ga^2\,. 
\end{equation*}
More precisely, under our assumptions on $\ga(s)$ two cases can occur
\begin{enumerate}
\item There does not exist a zero energy resonance\footnote{See section \ref{sec3} for the definition of zero energy resonance.} for $h_n$, then $h_n^\ve$ converges to the Laplacian on the graph with decoupling gluing conditions in the vertex.
\item There exists a zero energy resonance $f_{r,n}$, for $h_n$. In such a case one can define two real constants $c_{\pm} = \lim_{s \to \pm \infty} f_r$, such that $c_+^2+c_-^2=1$, and the limit operator on the graph is the Laplacian with gluing conditions given by
\begin{equation}
\label{scaleinv}
c_- f_1(0) = c_+ f_2(0)\,,\qquad
c_+ f_1' (0 ) +c_- f_2' (0)=0\,.
\end{equation}
\end{enumerate}
The convergence has to be intended in the norm resolvent sense.

We remark that the existence of a zero energy resonance is an exceptional event, that in our case is related to some very special choices of the  initial curve $\Gamma$, see \cite{ACF07} for few examples. This implies that in most of the cases the limit operator is defined by decoupling conditions in the vertex. For this reason we call the  case 1 \emph{generic} and the case 2 \emph{non-generic}.

The gluing conditions \eqref{scaleinv} are known in literature as scale invariant, see \cite{HC06}, and are parameterized by one independent real parameter, e.g., the ratio $c_+/c_-$. By using the result proved in \cite{CE07} we shall see that, in the non-generic case, a deformation of order $\ve$ of the angle $\theta$ between the edges of the graph leads to a more general class of gluing conditions, i.e., the ones defined by
\begin{equation*}
c_- f_1(0) = c_+ f_2(0 )\, ,\qquad
c_+ f_1' (0 ) +c_- f_2' (0)=\hat b(c_+ f_1(0)+c_- f_2(0))
\end{equation*}
where $\hat b$ is a real constant related to the deformation parameter.

It is worth noticing that in the non-generic case the gluing conditions in the vertex implies a coupling between the edges; our analysis includes a wide class of boundary conditions and holds for all transverse modes.

The first results in the same direction presented here, in a setting with several edges,  concerned the case of compact networks of tubes\footnote{For compact network we mean a network that  is contained in a compact region of the space. For such a network and for the corresponding underlying graph the spectrum of the Laplacian is discrete.} with Neumann boundary conditions. In \cite{FW93}  was proved the  convergence of the solutions of the heat equation on the network to the corresponding solution of the equation on the graph. The proof made use of some results on the convergence of Markov processes proved in the same paper. In a similar setting the convergence of the spectrum was proved in \cite{KZ01}, \cite{RS01}, \cite{Sai01}, \cite{KS03} and \cite{EP05}. In all these works the gluing conditions arising in the limit are of free type. The most recent result on the Neumann problem (and in a setting in which the approximating manifold has no boundary), was given in  \cite{Pos06}. In the latter work in the case of compact and non compact networks, was proved the strong resolvent convergence to the operator on the graph with  gluing conditions in the vertex  of free type. In all
the aforementioned papers, only the projection onto the lowest transverse mode was considered.

The case we just described is the most simple one. In this setting the energy of the lowest transverse mode is equal to zero and,  in a neighborhood  of the vertex, one is allowed to approximate the wave function on the network  with a constant function.
The case  with Dirichlet boundary conditions has revealed much  more tricky. This is due to the fact that in the latter case the energy of the transverse modes always  increases as the inverse squared width of the tubes and this forces to rescale the Hamiltonian by subtracting the divergent energy term. The scaled operator has  a finite number of eigenvalues which all diverge as the network squeezes. 

A first result on the problem with Dirichlet boundary was given in \cite{Pos05}. In a setting of a compact network it was proved that the limit operator on the graph is characterized by decoupling conditions in the vertex. It is worth noticing that in this work the decoupling was  obtained as a consequence of an ad hoc hypothesis on the volume of the manifold in a neighborhood of the vertex.

Recently D. Grieser \cite{Gri07pp} has proved that, for a large class of boundary conditions, generically the limit gluing conditions are of decoupling type. This was already argued  by S. Molchanov and B. Vainberg, see \cite{MV07pp}. In these works, for compact networks, the spectral convergence of the Laplacian on the network to the operator on the graph is studied; the approach is based on the analysis of the scattering problem associated to the network of tubes and makes use of the analytic properties of the resolvent of the Laplacian on the manifold. The special cases in which the coupling occurs are related to the existence of singularity of the resolvent at the thresholds of the energy of the transverse modes.

Earlier the existence of a non-decoupling limit in the Dirichlet case was proved in \cite{ACF07} for the same model discussed in this paper. 
Such a model was proposed for the first time as a prototype of a Dirichlet network collapsing onto a graph in \cite{DT06} where 
the generic case leading to decoupling conditions was discussed.

In view of applications the most relevant type of boundary conditions in the modeling of constrained quantum mechanical systems  is the  Dirichlet one while  Neumann boundary conditions arise mostly in the case of electromagnetic or acoustic waveguides. Robin conditions are used for example in numerical simulations to model the interface between semiconductors, see, e.g., \cite{Sel84}.

The interplay between geometry and boundary conditions in waveguides has been studied in many works (see, e.g., \cite{ES89}, \cite{DE95}, \cite{DK02}, \cite{BMT07}, \cite{FK08} and references therein). Most of them focus on the differences between Neumann and Dirichlet boundary. By changing the Robin constant one can continuously switch from Neumann to Dirichlet boundary conditions. For this reason we guess that the analysis carried on in this paper can help to gain a deeper understanding in this problem.  

Let us stress that we consider only ``symmetric'' waveguides, i.e., we take the same boundary conditions on the upper and lower boundary of the waveguide. Asymmetric boundary conditions would lead to a divergent term  in the transverse Hamiltonian of the order of the inverse of the width of the waveguide. Some results on the spectral properties of the Laplacian in asymmetric waveguides can be found in \cite{Kre08pp} and \cite{KK05}.

The present paper is structured as follows. In section \ref{sec1} we  define  the Laplacian with symmetric Robin boundary conditions on  the waveguide and introduce the correct scaling to get a family of waveguides that collapses onto a graph. In section \ref{sec3} as preliminary results, we discuss the spectral structure of the one dimensional Laplacian on a compact interval with Robin boundary conditions and we recall some results on the limit of Hamiltonians with short range scaled potentials taken from \cite{ACF07}. After that we state the main theorem. Section \ref{sec4} is devoted to the proof of the main theorem. In section \ref{sec5} we discuss the effect of small deformations of the relevant parameters of the problem. In this section we make  use of the result proved in \cite{CE07} where the same problem was studied in the case of a Dirichlet boundary. A section of conclusions and remarks closes the paper. The proofs of few technical estimates and of a resolvent formula that will be used to prove the main theorem are postponed in appendix \ref{appendice}.

\section{The model \label{sec1}}
\setcounter{equation}{0}

Let $\Ga$ be a curve in $\erre^2$ given in parametric form by $\Ga:= \{ ( \ga_1(s), \ga_2(s) ), s\in \erre \}$
and let us assume that it is parameterized by the arc length $s$, i.e. $ \ga'_1(s)^2+ \ga'_2(s)^2=1$. The curve $\Ga$ is completely defined up to isometries once the signed curvature $\ga$ is known
\begin{equation*}
\gamma(s):=\gamma'_{2}(s)\gamma''_{1}(s)
-\gamma'_{1}(s)\gamma''_{2}(s)\,;
\end{equation*}
\n the curvature radius of $\Ga$ in $s$
 is equal to the inverse of the modulus of the signed curvature.

We shall assume that $\ga(s)\in C_0^{\infty}(\RE)$, therefore $\Ga$ is a straight line
outside a compact region. We shall also assume that $\Ga$ has no self-intersections. Thus
$\Ga$ consists of two straight lines, $l_1$ and $l_2$, with the
origins, $O_1$ and $O_2$, connected by an infinitely smooth, non self-intersecting, curve $C$, running in a compact region. The integral of
$\ga$ gives the angle $\theta  $  
between  $l_1$ and $l_2$
\begin{equation}
\label{theta}
\theta=\int_\RE\ga(s)ds\,.
\end{equation}

Let us denote the open strip of width $2d$ around $\Ga$ by $\Omega$:
\begin{equation*}
\Omega:= \{ (x,y)\,\,s.t.\,\,  x= \ga_1(s) -  u \ga_2'(s), y= \ga_2(s) +  u \ga_1'(s), s\in\erre, u\in(-d,d) \}\,.
\end{equation*}
We assume $\sup_s |\ga(s)|d < 1$, in this way $(s,u)$ provide a global
system of coordinates in $\Omega$.

Let us define the sesquilinear form $\QQ_\Omega$ on $L^2(\Omega)\times L^2(\Omega)$ with domain
\begin{equation*}
\DD(\QQ_\Omega):=H^1(\Omega)\times H^1(\Omega)\,,
\end{equation*}
given by
\begin{equation*}
\QQ_\Omega[\varphi,\psi]:=\int_\Omega dx\,dy
\overline{\nabla \varphi}\,\nabla\psi \,.
\end{equation*}
It is well known that $\QQ_\Omega$ is closed and positive and that the associated self-adjoint operator is
the Laplacian in the domain $\Omega$ with Neumann boundary conditions.

Now we
consider the following perturbation of $\QQ_\Omega$ depending on $\al\in\RE$
\begin{equation*}
\QQ_\Omega^R[\varphi,\psi]:=\int_\Omega dx\,dy
\overline{\nabla \varphi}\,\nabla\psi \,+\al\int_{\partial\Omega}\big(\overline\varphi\,\psi\big)\big|_{\partial\Omega}\,dS
\end{equation*}
where $dS$ is the Lebesgue induced measure on $\partial \Omega$.
Then by Sobolev embedding theorems (see, e.g., \cite{Ada75}), $\QQ_\Omega^R$ is a small perturbation
of $\QQ_\Omega$ in the sense of quadratic forms and $\QQ_\Omega^R$ is closed and bounded from below on
\begin{equation*}
\DD(\QQ_\Omega^R):=H^1(\Omega)\times H^1(\Omega)\,.
\end{equation*}
Moreover $C_0^\infty(\RE^2)\times C_0^\infty(\RE^2)$ is a core for $\QQ_\Omega$ and $\QQ_\Omega^R$ (see, e.g., \cite{RSII} Th. X.17). 

We denote by $-\Delta^R_{\Omega}$ the self-adjoint operator associated to $\QQ_\Omega^R$. One can verify that the operator
$-\Delta^R_{\Omega}$ coincides with the Laplacian with Robin boundary conditions on
$\partial\Omega$, i.e.,  functions 
in $\DD(-\Delta_\Omega^R)$ belong to $H^2 (\Omega)$
and their trace on $\partial \Omega$ satisfies the boundary condition 
$\pd{\psi}{n}|_{\partial\Omega}+\al\psi|_{\partial \Omega}=0$.

We put
\begin{equation}
\label{DD}
\DD := \lf\{ \psi \in H^2(\Omega)\; s.t.\;\pd{\psi}{n}\bigg|_{\partial\Omega}+\al\psi\big|_{\partial \Omega}=0 \ri\}
\end{equation}
and let us recall that by the first representation theorem of quadratic forms, see \cite{Kat80}, we have
\begin{equation*}
\DD(-\Delta_\Omega^R):=
\lf\{\psi\in\DD(\QQ_\Omega^R)\;s.t.\;\exists\chi\in L^2(\Omega),\forall \varphi\in\DD(\QQ_\Omega^R)\,,\; \QQ_\Omega^R[\varphi,\psi]=(\varphi,\chi)_{L^2(\Omega)}\ri\}\,.
\end{equation*}
Integrating by parts we have immediately that $\DD \subset \DD(-\Delta_\Omega^R)$. Now we prove 
the reverse inclusion and then equality follows.
Let us assume that $\psi \in \DD(-\Delta_\Omega^R)$ then there exists $\chi\in L^2(\Omega)$ such that 
\begin{equation*}
\QQ_\Omega^R[\varphi,\psi]=(\varphi,\chi)_{L^2(\Omega)}
\end{equation*}
for all $ \varphi \in C_0^{\infty} (\Omega)$. For such a $\varphi$ we simply have 
\begin{equation*}
\QQ_\Omega^R[\varphi,\psi]= \int_{\Omega}  dx\,dy \overline{\nabla \varphi}\,\nabla\psi
\end{equation*}
and then
\begin{equation*}
\int_{\Omega} dx\,dy \;\overline{\varphi}\chi =
\int_{\Omega} dx\,dy\overline{\nabla \varphi}\,\nabla\psi=
\int_{\Omega} dx\,dy\;\overline{\varphi}\,(-\Delta\psi)
\end{equation*}
which implies that $-\Delta\psi \in L^2(\Omega)$ that is $\psi \in H^2(\Omega)$. 
Now we take $\varphi  \in\DD(\QQ_\Omega^R)$ and using Gauss-Green theorem, we find
\begin{equation*}
\QQ_\Omega^R[\varphi,\psi]=
\int_{\Omega}  dx\,dy \;\overline{\varphi}\,(-\Delta\psi) +
\int_{\partial \Omega} dS \; \overline{\varphi}\big|_{\partial \Omega}
\lf[
\pd{\psi}{n}\bigg|_{\partial\Omega}+\al\psi\big|_{\partial \Omega}
\ri]\,.
\end{equation*}
This is a bounded functional with respect the $L^2(\Omega)$ topology in $\varphi$ if and only 
if the boundary conditions $\pd{\psi}{n}|_{\partial\Omega}+\al\psi|_{\partial \Omega}=0$ hold
and then $\DD(-\Delta_\Omega^R) = \DD$.

In order to study the properties of $-\Delta_{\Omega}^R$ it is convenient to use the coordinates $(s,u)$
which belong to $ \Omega' = \erre \times (-d,d)$. The following proposition holds true.
\begin{proposition}
\label{proproblap}
For $\ga\in C_0^\infty(\RE)$, $-\Delta_\Omega^R$ is unitarily equivalent to the operator $H$
in $L^2(\Omega', ds\,du)$ defined by
\begin{equation}
\label{domH}
\DD(H):=\bigg\{\psi\in H^2(\Omega')\,s.t.\,
\pd{\psi}{u} (s ,  d ) + \al_1 (s) \psi (s , d ) =0\;,
-\pd{\psi}{u} (s , - d ) + \al_2 (s) \psi (s , -d ) =0
\bigg\}
\end{equation}
where
\begin{equation}
\label{al1al2}
\al_1(s):=\al-\frac{\ga(s)}{2(1+d\ga(s))}\;,\quad
\al_2(s):=\al+\frac{\ga(s)}{2(1- d\ga(s))}\,,
\end{equation}
and
\begin{equation*} 
H:=
-\pd{}{s} \frac{1}{(1+u\gamma(s))^2}\pd{}{s}-
\pd{^2}{u^2}+V(s,u)
\end{equation*}
\n with
\begin{equation}
 \label{potV}
V(s,u):=-\frac{\gamma(s)^2}{4(1+u\gamma(s))^2}
+\frac{u\gamma''(s)}{2(1+u\gamma(s))^3}
-\frac{5}{4}\frac{u^2\gamma'(s)^2}{(1+u\gamma(s))^4}\,.
\end{equation}
\end{proposition}
\begin{proof}
We denote by $\widetilde{\QQ}^R_{\Omega}$ the sesquilinear form obtained by restricting $\QQ_\Omega^R$ to $C_0^\infty(\RE^2)\times C_0^\infty(\RE^2)$; the closure of $\widetilde{\QQ}^R_{\Omega}$ 
is $\QQ_\Omega^R$. The Hilbert space $L^2(\Omega, dx\,dy)$ is mapped by the change of variables 
into $L^2(\Omega',(1+u\ga)ds\,du)$. With these coordinates the form $\widetilde{\QQ}^R_{\Omega}$ reads
\begin{equation*}
\begin{aligned}
\widetilde{\QQ}^R_{\Omega}[\varphi,\psi]=&
\int_{\Omega'}\bigg(
\frac{1}{1+u \ga} \overline{\pd{\varphi}{s}}\pd{\psi}{s}+
(1+u \ga )\overline{\pd{\varphi}{u}}\pd{\psi}{u}\bigg) ds\,du \\
&+\al\int_\RE\big[(1+d \ga(s))\overline{\varphi}(s,d)\psi(s,d)+(1-d \ga(s))\overline{\varphi}(s,-d)\psi(s,-d))\big]ds\,.
\end{aligned}
\end{equation*}
Notice that even in the new coordinates the domain of $\widetilde \QQ^R_\Omega$ is  $C_0^\infty(\RE^2)\times C_0^\infty(\RE^2)$. We consider the unitary map 
$U:L^2(\Omega',(1+u\ga)ds\,du) \to L^2(\Omega', ds\,du)$ given by
\begin{equation}
\label{unitarymap}
(U\psi)(s,u):=(1+u\ga(s))^{1/2}\,\psi(s,u)\,.
\end{equation}
A straightforward calculation shows that the  form $\widetilde{\QQ}^R_{\Omega}$ is unitarily equivalent to the form $\widetilde{Q}^R_{\Omega'}$ in 
$L^2(\Omega', ds\,du)$ defined by
\begin{equation*}
\DD(\widetilde{Q}^R_{\Omega'}):=C_0^\infty(\RE^2)\times C_0^\infty(\RE^2)\,,
\end{equation*}
\begin{equation*}
\begin{aligned}
\widetilde{Q}^R_{\Omega'}[\varphi,\psi]:=&
\int_{\Omega'}\bigg(
\frac{1}{(1+u \ga)^2} \overline{\pd{\varphi}{s}}\pd{\psi}{s}+
\overline{\pd{\varphi}{u}}\pd{\psi}{u}+V\,\overline{\varphi}\,\psi\bigg) ds\,du \\
&+\int_\RE(\al_1(s)\overline{\varphi}(s,d)\psi(s,d)+\al_2(s)\overline{\varphi}(s,-d)\psi(s,-d))ds
\end{aligned}
\end{equation*}
where $\al_1(s)$ and $\al_2(s)$  are given by equation \eqref{al1al2} and $V$ by equation \eqref{potV}. Now we just need to compute 
$Q^R_{\Omega'}$, the closure of $\widetilde{Q}^R_{\Omega'}$. Let us consider
\begin{equation*}
\widetilde{Q}_{\Omega'}[\varphi,\psi]=
\int_{\Omega'}\bigg(
\frac{1}{(1+u \ga)^2} \overline{\pd{\varphi}{s}}\pd{\psi}{s}+
\overline{\pd{\varphi}{u}}\pd{\psi}{u} \bigg)ds\,du
\end{equation*}
with the same domain as $\widetilde Q^R_{\Omega'}$. Since $0<c <(1+ d \ga) < c^{-1}$ for some positive constant $c$, we notice that $\widetilde{Q}_{\Omega'}[\psi,\psi]$ is equivalent to the $H^1$ norm
in $\Omega'$. Since $V$ and $\al_i$ are bounded, then  
$\widetilde{Q}^R_{\Omega'}$ is a small perturbation
of $\widetilde{Q}_{\Omega'}$ and its closure is given by
\begin{equation*}
\DD({Q}^R_{\Omega'}):=H^1(\Omega')\times H^1(\Omega')
\end{equation*}
\begin{equation*}
\begin{aligned}
{Q}^R_{\Omega'}[\varphi,\psi]:=&
\int_{\Omega'}\bigg(
\frac{1}{(1+u \ga)^2} \overline{\pd{\varphi}{s}}\pd{\psi}{s}+
\overline{\pd{\varphi}{u}}\pd{\psi}{u}+V\,\overline{\varphi}\,\psi\bigg) ds\,du \\
&+\int_\RE(\al_1(s)\overline{\varphi}(s,d)\psi(s,d)+\al_2(s)\overline{\varphi}(s,-d)\psi(s,-d))ds\,.
\end{aligned}
\end{equation*}
It is easy to repeat the argument  used before  to prove that  the domain of $-\Delta_{\Omega}^R$ is equal to $\mathscr D$ defined in \eqref{DD}, to see that the self-adjoint operator associated with ${Q}^R_{\Omega'}$ coincides with the operator $H$ given by  \eqref{domH} - \eqref{potV}.
\end{proof}

From now on we denote $L^2(\Omega', ds\,du)$ simply by $L^2$.

In order to get a family of waveguides that collapses onto a prototypical graph, we rescale the geometric parameters of the system $\ga$ and $d$ in the following way:
\begin{align}
\ga(s)& \longrightarrow \,\,\frac{1}{\ve} \ga\lf( \frac{s}{\ve} \ri) 
\label{scalinggamma1}
\\
d &\longrightarrow \,\,\de^{\ve} d \,.
\label{scalingd}
\end{align}
Where $\ve>0$ and $\de^\ve>0$ are two dimensionless scaling  parameters,  such that $\de^\ve\to 0$ as $\ve\to 0$. With this scaling we have two characteristic scales: $\ve$ is the scale of variation of the curvature $\ga$ 
while  $\de^\ve$ is the intrinsic scale  in the transverse direction.
We shall consider the regime where $\de^\ve \ll \ve$, that is, the curvature is slowly
varying with respect to the width of the waveguide. In particular  we assume $\de^\ve:=\ve^a$ with $a\geqslant 1$ big enough, a more precise statement on the ratio $\de^\ve/\ve$ will be given in the following. Notice that the angle $\theta$ defined in \eqref{theta} is unchanged by  the scaling \eqref{scalinggamma1}.

We also rescale the Robin constant $\al$ in the following way
\begin{equation*}
\al \longrightarrow \al^\ve:=\frac{\al}{\de^\ve}\,,
\end{equation*}
in this way the Robin boundary conditions in \eqref{DD} are invariant under the scaling of $d$.

We obtain a family of domains $\Omega^{\ve}$ and of operators $-\Delta^R_{\Omega^{\ve}}$
such that $\Omega^{\ve}$ approximates, for $\ve \rightarrow 0$, the
broken line of angle $\theta$ made up of two half lines, $l_1$ and
$l_2$, with the same origin, $O_1\equiv O_2\equiv O$. Since we assume $a \geqslant 1$
then $(s,u)$ are a system of global coordinates also for $\Omega^{\ve}$.

From proposition \ref{proproblap} it  follows that, for all $\ve>0$, the operator $-\Delta^R_{\Omega^{\ve}}$ is unitarily equivalent
to the self-adjoint operator $H^{\ve}:\DD(H^\ve)\subset L^2\to L^2$ given by
\begin{equation}
\label{Hve}
H^{\ve}:=-
\pd{}{s} \frac{1}{(1+u \eta^\ve(s))^2}\pd{}{s}-
 \frac{1}{\de^{\ve\, 2}} \pd{^2}{u^2} +\frac{1}{\ve^2} V^\ve(s,u)\,,
\end{equation}
\n with
\begin{equation*}
V^\ve(s,u):=-\frac{\gamma(s/\ve)^2}{4(1+u \eta^\ve(s))^2}
+\frac{\de^\ve/\ve \, u\gamma''(s/\ve)}{2(1+u \eta^\ve(s))^3}
-\frac{5}{4}\frac{(\de^\ve / \ve)^2 u^2\gamma'(s/\ve)^2}{(1+u \eta^\ve(s))^4}
\end{equation*}
and with domain
\begin{equation}
\label{robinstorto}
\DD(H^\ve):=\bigg\{\psi\in H^2(\Omega')\,s.t.\,
 \pd{\psi}{u} (s , d ) + \al^{\ve}_1 (s) \psi (s ,d ) =0\;,
-\pd{\psi}{u} (s ,- d ) + \al^{\ve}_2 (s) \psi (s , -d ) =0
\bigg\}\,,
\end{equation}
here  $\al_1^\ve (s)$ and $\al_2^\ve (s)$ are given by
\begin{equation}
\label{alves}
\al_1^\ve (s):=\al-\frac{\eta^\ve (s)}{2(1+d \eta^\ve(s))}\;,\quad
\al_2^\ve (s):=\al+\frac{\eta^\ve (s)}{2(1- d \eta^\ve(s))}
\end{equation}
and we have introduced $\eta^\ve(s): = \de^\ve / \ve \, \ga(s /\ve)$.  In the following
$\eta^\ve$ will play the role of a small quantity in a suitable topology  and will allow a perturbative
analysis.

\section{Main result\label{sec3}}
\setcounter{equation}{0}

Equation \eqref{Hve} shows that the transverse kinetic energy is divergent in the limit $\ve \to 0$.
This is a common problem for these kind of singular limits. In order to overcome this problem, it is convenient
to introduce an $s$ dependent orthonormal complete set of states of $L^2((-d,d))$ which diagonalizes   the transverse
part of the Hamiltonian and provides a useful framework to discuss the limit of $H^\ve$ in the sense roughly described in the introduction. To this aim we start this section  with a short discussion on the one dimensional Robin Laplacian in $L^2((-d,d))$.

Given two  real constants $\al_1$ and $\al_2$, we denote by $\hhd_{\al_1 , \al_2}$ the Robin Laplacian on $L^2((-d,d))$; $\hhd_{\al_1 , \al_2}$ is the self-adjoint operator defined as
\begin{equation*}
\DD(\hhd_{\al_1, \al_2}):=\{\psi\in H^2((-d,d))\,s.t.\,\,
\psi'(d)+\al_1 \psi(d)=0\; , \,
-\psi'(-d)+\al_2 \psi(-d)=0 \}
\end{equation*}
\begin{equation*}
\hhd_{\al_1, \al_2}\psi:=-\frac{d^2\psi}{du^2}\quad\forall\psi\in\DD(\hhd_{\al_1, \al_2})\,.
\end{equation*}
Let us denote by $g_{\al_1, \al_2}(k^2):=( \hhd_{\al_1 , \al_2 } - k^2 )^{-1}$ the resolvent of $\hhd_{\al_1 , \al_2 }$.
The integral kernel of $g_{\al_1, \al_2}(k^2)$ is explicitly known:
\begin{equation}
\label{reso}
\begin{aligned}
g_{\al_1, \al_2}(k^2;u,u')=&  \frac{1}{2k d^2}\frac{\sin[k(2d-|u-u'|)]}{\cos(2kd)}\\
&-\frac{k(\al_1-\al_2)\sin[k(u+u')]-(\al_1\al_2+k^2)\cos[k(u+u')]}{2k d^2
[(\al_1\al_2-k^2)\sin(2kd)+k(\al_1+\al_2)\cos(2kd)]}\\
&-\frac{(\al_1\al_2-k^2)\cos[k(u-u')]}{2k d^2\cos(2kd)[(\al_1\al_2-k^2 )\sin(2kd)+k(\al_1+\al_2)\cos(2kd)]}
\end{aligned}
\end{equation}
for  $k^2\in\rho(\hhd_{\al_1,\al_2})$  and $\Im k\geqslant 0$ where $\rho(\hhd_{\al_1,\al_2})$ denotes the resolvent set of $\hhd_{\al_1,\al_2}$.\\

We denote by $\la_n$, $n=0,1,2, \ldots$, the eigenvalues of $\hhd_{\al_1, \al_2}$ arranged in  increasing order. Using \eqref{reso} it is straightforward  to prove that $\la_n=k_n^2$ with $k_n$ given by the solutions of 
\begin{equation}
(\al_1\al_2-k_n^2)\sin(2k_n d)+k_n(\al_1+\al_2)\cos(2k_n d)=0,
\label{eigen}
\end{equation}
positive eigenvalues correspond to $k_n\in\RE^+$, while  negative eigenvalues are given by $k_n\in i\RE^+$.  The corresponding eigenfunctions have the form
\begin{equation}
\label{phin}
\phi_n(u)=A_n \sin(k_n u)+B_n \cos(k_n u)\qquad  n=0,1,2,\dots\,,
\end{equation}
where $A_n$ and $B_n$ are suitable coefficients. The eigenfunctions $\phi_n$ can be chosen real. When we want to stress the dependence on the boundary conditions,  we shall denote the eigenvectors of $\hhd_{\al_1, \al_2}$ by $\phi_n(\underline{\al} , u)$, where $\underline{\al} \equiv (\al_1 , \al_2)$.\\

We simply denote $\hhd_{\al, \al}$ by $\hhd_{\al}$ and by $\mu_n$, $n=0,1,2, \ldots$, its eigenvalues
 arranged in increasing order. The eigenvalues of $\hhd_\al$ can be written as $\mu_n = p_n^2$ with $p_n$ satisfying:
\begin{align}
\label{eigvaleven}
&p_n \sin (p_n d) -\al\cos (p_n d)=0 \qquad n=0,2,4,\dots\\
\label{eigvalodd}
&p_n \cos ( p_n d) +\al\sin (p_n d) =0\qquad n=1,3,5,\dots\,.
\end{align}
 Positive eigenvalues correspond to $p_n\in\RE^+$, while  negative eigenvalues are given by $p_n\in i\RE^+$.
 The corresponding eigenvectors have now a definite parity and can be written as
\begin{align}
&\xi_n(u)=N_n\cos(p_n u)\qquad  n=0,2,4,\dots
\label{xin1}
\\
&\xi_n(u)=N_n\sin(p_n u)\qquad  n=1,3,5,\dots\,,
\label{xin2}
\end{align}
where $N_n$ is the normalization constant.
For $\al \geqslant 0$ all the eigenvalues are non negative, for $-1 \leqslant \al d < 0$ there is one negative eigenvalue and for $ \al d <-1$ there are two negative eigenvalues. In figure \ref{fig1} the first four eigenvalues of $\hhd_\al$ are plotted as functions of $\al$ (for $d=1$). 
\begin{figure}[h!]
\begin{center}
\includegraphics[width=8cm,angle=-90]{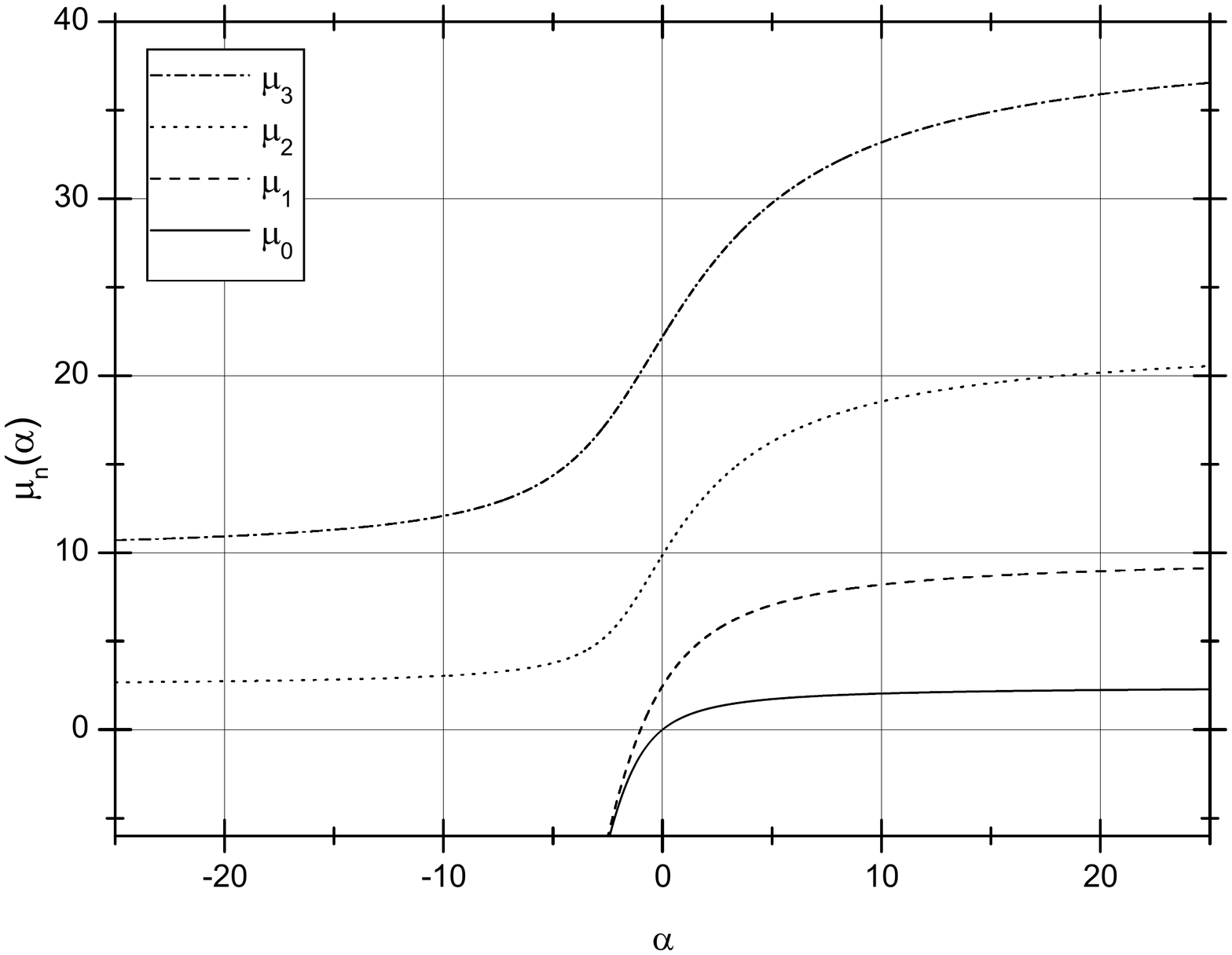}
\caption{
\label{fig1}
Plot of the first four eigenvalues of $\hhd_\al$ as functions of $\al$, in the plot it is assumed $d=1$.}
\end{center}
\end{figure}

Let us now take $\al_1$ and $\al_2$ in $\hhd_{\al_1,\al_2}$ of the following form:
\begin{equation*}
\al_1:=\al-\frac{\eta}{2(1+d\eta)}\;,\quad
\al_2:=\al+\frac{\eta}{2(1- d\eta)}\,.
\end{equation*}
For $\eta \ll 1$ the Hamiltonian $\hhd_{\al_1,\al_2}$ can be considered as a perturbation of $\hhd_\al$.
The eigenvalues of $\hhd_{\al_1,\al_2}$, $\la_n=k_n^2$, are  defined by the equation $\Delta(k_n, \eta)=0$ with
\begin{equation*}
\begin{aligned}
\De(k_n,\eta):=&
\bigg[\bigg(\al+\frac{\eta}{2(1-d\eta)}\bigg)\bigg(\al-\frac{\eta}{2(1+d\eta)}\bigg)-k_n^2\bigg]\sin(2k_n d)\\
&+k_n\bigg[\bigg(\al+\frac{\eta}{2(1-d\eta)}\bigg)+\bigg(\al-\frac{\eta}{2(1+d\eta)}\bigg)\bigg]\cos(2k_n d)\,.
\end{aligned}
\end{equation*}
We want to give a perturbative
expansion of the eigenvalues $\la_n$ in the small parameter $\eta$ up to the second order:
\begin{equation*}
\la_n= \sum_{j=0}^{2} \la_n^{(j)} (d \eta)^j +{\mathcal O}(\eta^3)\,.
\end{equation*}
The length $d$ appears in the expansion for dimensional reasons only, the small parameter here is $\eta$.
The coefficients $k_n^{(j)}$, for $j=0,1,2$, in the  expansion $k_n=k_n^{(0)}+k_n^{(1)} d \eta+k_n^{(2)}(d \eta)^2+\OO(\eta^3)$ can be obtained from the expansion
$\De(k_n,\eta)=\De^{(0)}+\De^{(1)}d\eta+\De^{(2)}(d \eta)^2+\OO(\eta^3)$ and imposing $\De^{(j)}=0$,  $j=0,1,2$. A straightforward calculations gives
\begin{equation*}
\begin{aligned}
k^{(0)}_n & = p_n \\
k^{(1)}_n & = 0 \\
k^{(2)}_n & = -\frac{p_n\big[\al-2d\big(\al^2+p_n^2\big)\big]}{4d^2\big(\al^2+p_n^2\big)\big[\al+d \big(\al^2+p_n^2\big)\big]} \\
\end{aligned}
\end{equation*}
where $p_n$ were defined in  equations \eqref{eigvaleven} and \eqref{eigvalodd}, and since $\la_n = k_n^2$ we immediately obtain
\begin{equation}
\begin{aligned}
\la^{(0)}_n & = \mu_n \\
\la^{(1)}_n & = 0 \\
\la^{(2)}_n & = 
-\frac{\mu_n[\al-2d(\al^2+\mu_n)]}{2d^2(\al^2+\mu_n)[\al+d(\al^2+\mu_n)]}\,.
\label{lan2}
\end{aligned}
\end{equation}
The  coefficients $\la_n^{(2)}$, $n=0, 1, \ldots $, are implicitly defined as functions of $\al$. 

Equations \eqref{Hve}, \eqref{robinstorto} and \eqref{alves} indicate that 
for each $s$ the transverse kinetic term of $H^\ve$ has the same form of $\hhd_{\al_1,\al_2}$ and that
our hypothesis allows a perturbative analysis of its spectrum.  Therefore for all $s\in\RE$ we consider the  Hamiltonian $\hhd_{\al_1^\ve,\al_2^\ve}:L^2((-d,d))\to L^2((-d,d))$, where $\al_1^\ve(s)$ and  $\al_2^\ve(s)$ were defined in \eqref{alves}. Such Hamiltonian depends on $\ve$ and $s$ only via the constants $\al_1^\ve(s)$ and $\al_2^\ve(s)$ in the  boundary conditions. In the following we shall use the notation $\underline{\al}^\ve(s)\equiv(\al_1^\ve(s),\al_2^\ve(s))$.  

The eigenvalues and  the eigenvectors of  $\hhd_{\al_1^\ve,\al_2^\ve}$ are defined  according to \eqref{eigen} and \eqref{phin}  and will be denoted by $\la_n^\ve(s)$ and  $\phi_n^\ve(s)$. We shall use also the notation $\phi_n^\ve(s)\equiv\phi_n( \underline{\al}^\ve(s))$ to remind that $\phi_n^\ve(s)$ depend on $s$ and $\ve$ only via $\al_1^\ve(s)$ and $\al_2^\ve(s)$. For fixed $s\in\RE$, $\big\{\phi_n( \underline{\al}^\ve(s) )\big\}_{n=0,1,\dots}$ is an orthonormal basis of $L^2((-d,d))$.

For all $z\in{\CO\backslash\RE}$ and for all $m,n=0,1,2,\dots$, we denote by $r_{m,n}^\ve(z)$  the reduced resolvent of $H^\ve$, i.e., the operator in $\BB\big(L^2(\RE),L^2(\RE)\big)$ defined via its integral kernel by\footnote{We denote by $\mathscr{B}(L^2(\HH),L^2(\HH'))$ the Banach space of bounded operators from $\HH\to\HH'$ and by $\|\cdot\|_{\mathscr{B}(L^2(\HH),L^2(\HH'))}$ the corresponding norm.}
\begin{equation*}
r_{m,n}^\ve(z;s,s'):=
\int_{-d}^{d}\int_{-d}^{d}
\phi_m(\underline\al^\ve (s),u)\,\Big(H^\ve-\frac{\mu_n}{\dde}-z\Big)^{-1}(s,u;s',u')\,\phi_n(\underline\al^\ve (s'),u')
du\,du'\,.
\end{equation*}
Notice that we have subtracted the divergent quantity $\mu_n/\dde$ from $H^\ve$ in order 
to compensate the divergence of the transverse kinetic energy and get a non trivial limit;
this procedure was already used in \cite{Pos05}, \cite{DT06}, \cite{ACF07} and \cite{CE07}.\\

We also need to recall some facts on one dimensional Schr\"odinger operators with short range potentials in $L^2(\RE)$.
Let us consider the Hamiltonian $h$ given by:
\begin{equation}
\label{barh}
h:=-\frac{d^2}{ds^2}+ v(s)\,,
\end{equation}
and let us assume that for some $c>0$
\begin{equation}
\label{vass}
\int_\erre  v(s)ds\neq0\qquad\qquad\qquad
e^{c|\cdot\,|} v\in L^1(\erre)\,.
\end{equation} 
We say that $ h$ has a zero
energy resonance if there exists $f_r\in L^{\infty}(\erre)$,  
$f_r \notin L^2(\erre)$ such that $hf_r =0$ in distributional sense.
Furthermore, if $f_r$ exists, it is unique, up to a trivial
multiplicative constant and one can define two constants 
\begin{equation}
\label{c1c2}
c_-:=\lim_{s \to - \infty} f_r (s)\qquad\textrm{and}\qquad  c_+:=\lim_{s \to+\infty} f_r (s)\,.
\end{equation}
The constants $c_-$
and $c_+$ can not be both zero, in such a case $f_r$
would be in $L^2(\erre)$, then zero would be an eigenvalue for $h$
(see Lemma 2.2. in \cite{BGW85}), but this is impossible under our
assumptions on $v$, see Theorem 5.2. in \cite{JN01}.
We can choose  $c_-$ and $c_+$ real and such that $c_-^2+c_+^2=1$.

Let $h_r$ be the following family of self-adjoint operators depending on $c_-$ and $c_+$
\begin{equation}
\label{domhr}
{\mathscr D}(h_r) := \{ f\in H^2(\erre \setminus 0 ) \,\, s.t. \,\,
c_- f(0^+ ) = c_+ f(0^- )
\, , \, c_+ f' (0^+ ) -c_- f' (0^- )=0 \}
\end{equation}
\begin{equation}
\label{hr}
h_r f := - \frac{d^2 f}{ds^2} \qquad s\neq 0\,.
\end{equation}
The Hamiltonian $h_r$ is a self-adjoint extension of the
symmetric operator $-\Delta$ in dimension one  defined on
$C_0^{\infty}(\erre \setminus \{ 0 \} )$.  For $c_-=c_+$ the operator $h_r$ coincides with the free   Laplacian on the line; we refer to \cite{ABD95} for
a comprehensive characterization of the point perturbations of the
Laplacian in dimension one. Let us notice that the operator $h_r$ can be rewritten as an operator on $H^2((0,\infty))\oplus H^2((0,\infty))\subset L^2((0,\infty))\oplus L^2((0,\infty))$ by defining $f_1$ and $f_2$ in $H^2((0,\infty))$ such that $f_1(x)=f(x)$ for $x>0$ and $f_{2}(-x)=f(x)$ for $x<0$. Within this notation the condition in $x=0$ in the domain of $h_r$ can be written as we did in the introduction, see equation \eqref{scaleinv}.
 
We denote the one dimensional Laplacian with decoupling (or  Dirichlet) gluing  conditions in the origin by $h_0$
\begin{equation*}
{\mathscr D}(h_0) := \{ f\in H^2(\erre \setminus 0 ) \cap
H^1(\erre  )\,\,s.t.\,\,  f(0)=0 \} 
\end{equation*}
\begin{equation*}
h_0 f: = - \frac{d^2 f}{ds^2} \qquad s\neq 0\,.
\end{equation*}
Even the operator $h_0$ can be written in the standard notation of the Laplacian  on graphs, the condition in $x=0$ would correspond to a decoupling condition in the vertex, see equation \eqref{deccon}.

Now we rescale $h$ in the following way 
\begin{equation}
\label{barhve}
 h^\ve:=-\frac{d^2}{ds^2}+\frac{1}{\ve^2} v(s/\ve)
\end{equation}
and we discuss the convergence of $h^\ve$ in resolvent sense. The following proposition is taken from lemma 1 in \cite{ACF07}.
\begin{proposition}
\label{prop1}
Take $h$ and $h^\ve$ defined as above and
assume \eqref{vass}. Then two cases can occur:
\begin{enumerate}
\item There does not exist a zero energy resonance for the Hamiltonian $h$, then
\begin{equation*}
\ulim_{\ve\to0} ( h^\ve - z )^{-1} = ( h_0 -z )^{-1} \qquad
 z \in \CO \setminus \erre\,.
\end{equation*}
\item There exists a zero energy resonance $f_r$ for the Hamiltonian $h$, then
\begin{equation*}
\ulim_{\ve\to0} (h^\ve - z )^{-1} = (h_r -z )^{-1} \qquad
z \in \CO \setminus \erre\,
\end{equation*}
where $h_r$ was defined in \eqref{domhr} and \eqref{hr}.
\end{enumerate}
\end{proposition}

 We denote by $\bt_n$ the coefficients
 \begin{equation}
 \label{btn}
 \bt_n:=-1/4+\la_n^{(2)}
 \end{equation}
  and by $h_n$ the Hamiltonian
\begin{equation*}
h_n:=-\displaystyle\frac{d^2}{ds^2}+\bt_n\ga^2(s)\,.
\end{equation*}
Our main result is stated in the  following theorem.
\begin{theorem}
\label{mainth}
Assume that $\Ga$ has no self-intersections and that $\ga\in C_0^\infty(\RE)$, moreover take $a>3 $, then for all $n,m=0,1,2,\dots$ two cases can occur:
\begin{enumerate}
\item For all $n$ such that there does not exist a zero energy resonance for $h_n$ we have
\begin{equation*}
\ulim_{\ve\to0}r_{m,n}^\ve(z)= \delta_{m,n}( h_0 -z )^{-1} \qquad
 z \in \CO \setminus \erre\,.
\end{equation*}
\item For all $n$ such that there exists  a zero energy resonance, $f_{r,n}$, for  $h_n$ we have
\begin{equation*}
\ulim_{\ve\to0} r_{m,n}^\ve(z)=\delta_{m,n} (h_{r,n} -z )^{-1} \qquad
z \in \CO \setminus \erre\,.
\end{equation*}
where   $h_{r,n}$ is defined according to equations \eqref{domhr} and \eqref{hr}.
\end{enumerate}
\end{theorem}
Notice that the constants $c_{-,n}$ and $c_{+,n}$ in the definition of $\DD(h_{r,n})$ may depend on $n$ and are related to $f_{r,n}$ via equation \eqref{c1c2}. 

We remark that the parameter $\al$ in the  boundary conditions affects the limit only through the coefficients $\beta_n (\al)$.  In figure \ref{fig2} the  functions  $\bt_n(\al)$ are plotted for $n=0,1,2,3$ and $d=1$.
\begin{figure}[h!]
\begin{center}
\includegraphics[width=8cm,angle=-90]{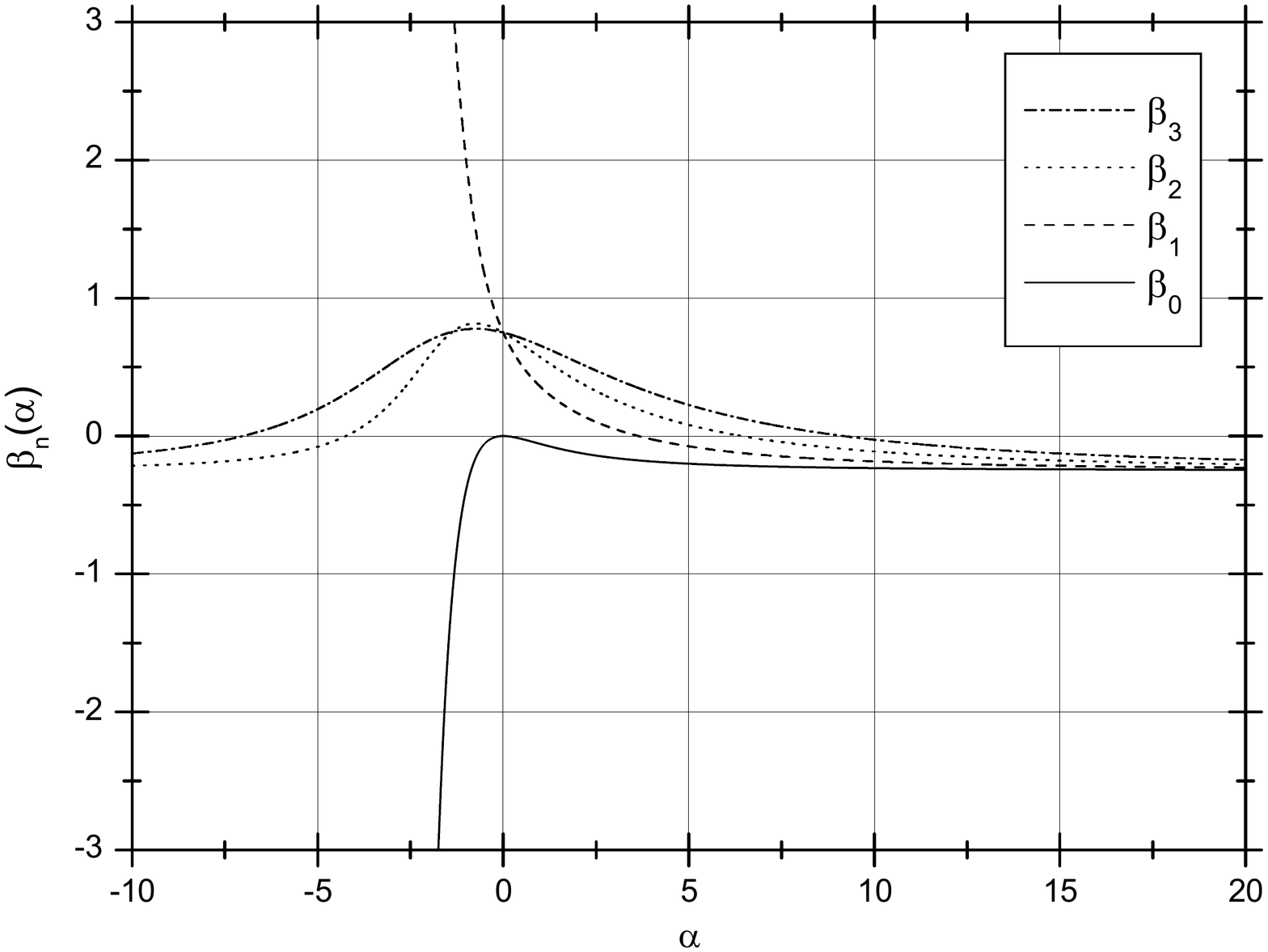}
\caption{
\label{fig2}
Plot of $\beta_n(\al)$ for $n=0,1,2,3$ and $d=1$.}
\end{center}
\end{figure}
The exceptional case $\bt_n(\al)=0$ is included in the statement of the theorem because in such a case the Hamiltonian $h_n$ has a zero energy resonance given by the constant function, therefore $c_{-,n}$ and $c_{+,n}$ coincide and the limit operator $h_{r,n}$ is the free, one dimensional, Laplacian.

\section{Proof of theorem \ref{mainth}\label{sec4}}
\setcounter{equation}{0}

Before getting into the technical core of this section devoted to the the proof of the main theorem, let us spend few words on the strategy we shall follow. In our regime the curvature is slowly varying with respect to the transverse dynamics. In
particular this means that $\al^\ve_1 (s)$ and $\al^\ve_2 (s)$ are slowly varying functions and that each subspace
corresponding to an eigenstate $\phi_n^\ve$ is ``adiabatically protected'' in the limit $\ve \to 0$.
This allows us to split the proof into two steps.

\begin{enumerate}
\item  First we prove that each subspace
corresponding to an eigenstate $\phi_n^\ve$ is adiabatically protected in the limit $\ve \to 0$ and we prove that the leading term in the reduced Hamiltonian (up to the renormalization factor $\mu_n/\de^{\ve\,2}$) is 
\begin{equation}
\label{hven}
\hh^\ve_n:=-\frac{d^2}{ds^2}+\frac{\bt_n}{\ve^2}\ga^2(s/\ve)
\end{equation}
where $\bt_n$ were defined in \eqref{btn}. This is done in lemma \ref{lemma1} and lemma \ref{lemma2}.

\item As a second step we  study the limit of $h^\ve_n$. Here we shall make use of the  proposition   \ref{prop1} to prove that, for each $n$, two cases can occur: if the potential $\beta_n\gamma^2$ generates a zero energy resonance for $h_n$,  then $h_n^\ve$ converges to an operator of  the family defined in \eqref{domhr} - \eqref{hr} otherwise the limit operator  is $h_0$, i.e.,  the Laplacian on the line with decoupling conditions in the origin.
\end{enumerate}

In the proofs $c$ will denote a generic positive constant whose value can change from line to line. 

Let $\hat H^\ve$ be the Hamiltonian 
\begin{equation}
\label{domhami}
\DD(\Hd^\ve):=\bigg\{\psi\in H^2(\Omega')\,s.t.\,
\pd{\psi}{u} (s ,  d ) + \al_1^\ve (s) \psi (s , d ) =0\;,
-\pd{\psi}{u} (s , - d ) + \al_2^\ve (s) \psi (s , -d ) =0
\bigg\}
\end{equation}
\begin{equation}
\Hd^\ve:=-\pd{^2}{s^2}-\frac{1}{\de^{2\,\ve}}\pd{^2}{u^2}-\frac{1}{\ve^2}\frac{\ga^2(s/\ve)}{4}\,.
\label{hami}
\end{equation}
For all $z\in{\CO\backslash\RE}$ and for all $m,n=0,1,2,\dots$, we denote by $\rbar_{m,n}^\ve(z)$  the reduced resolvent of $\Hd^\ve$, i.e., the operator in $\BB\big(L^2(\RE),L^2(\RE)\big)$ defined via its integral kernel by
\begin{equation*}
\hat r_{m,n}^\ve(z;s,s'):=
\int_{-d}^{d}\int_{-d}^{d}
\phi_m(\underline\al^\ve (s),u)\,\Big(\Hd^\ve-\frac{\mu_n}{\dde}-z\Big)^{-1}(s,u;s',u')\,\phi_n(\underline\al^\ve (s'),u')
du\,du'\,.
\end{equation*}

In the following lemma we  prove that $\hat r_{m,n}^\ve$ approximates $ r_{m,n}^\ve$.
\begin{lemma}
\label{lemma1}
Let $\ga\in C_0^\infty(\RE)$ and  $a>3$ then for all $n,m=0,1,2,\dots$
\begin{equation*}
\ulim_{\ve\to 0}
\big(r_{m,n}^\ve(z)-\hat r_{m,n}^\ve(z)\big)=0\qquad \forall z\in\CO\backslash\RE\,.
\end{equation*}
\end{lemma}
\begin{proof}
It is sufficient to prove that there exists $\ve_0$ such that for all $0<\ve<\ve_0$ and for all $f,g\in C_0^\infty(\RE)$
\begin{equation}
\label{goal}
\big|\big(g,\big[r_{m,n}^\ve(z)-\hat r_{m,n}^\ve(z)\big]f\big)_{L^2(\RE)}\big|\leqslant
c\ve^{a-3} \|g\|_{L^2(\RE)} \|f\|_{L^2(\RE)}    \,. 
\end{equation}
We shall make  use of the fact that given a self-adjoint operator $A$ in some Hilbert space $\mathcal H$ the inequality
\begin{equation}
\label{stimaresolvent}
\|(A-z)^{-1}\|_{\mathscr{B}(\mathcal H,\mathcal H)}\leqslant\frac{1}{|\Im z|}
\end{equation}
holds. From the first resolvent identity and from the definition of $r_{m,n}^\ve(z)$ and $\hat r_{m,n}^\ve(z)$ we have that
\begin{equation*}
\begin{aligned}
\big|\big(g,\big[r_{m,n}^\ve(z)-\hat r_{m,n}^\ve(z)\big]f\big)_{L^2(\RE)}\big|& \leqslant
\bigg|\bigg(g\phi_m^\ve,\Big(H^\ve-\frac{\mu_n}{\de^{\ve\,2}}-z\Big)^{-1}
b_1^\ve\pd{^2}{s^2}
\Big(\hat H^\ve-\frac{\mu_n}{\de^{\ve\,2}}-z\Big)^{-1}f\phi_n^\ve\bigg)_{L^2}\bigg| \\
& + \bigg|\bigg(g\phi_m^\ve,\Big(H^\ve-\frac{\mu_n}{\de^{\ve\,2}}-z\Big)^{-1}
b_2^\ve\pd{}{s}
\Big(\hat H^\ve-\frac{\mu_n}{\de^{\ve\,2}}-z\Big)^{-1}f\phi_n^\ve\bigg)_{L^2}\bigg| \\
& +\bigg|\bigg(g\phi_m^\ve,\Big(H^\ve-\frac{\mu_n}{\de^{\ve\,2}}-z\Big)^{-1}
\frac{1}{\ve^2}W^\ve
\Big(\hat H^\ve-\frac{\mu_n}{\de^{\ve\,2}}-z\Big)^{-1}f\phi_n^\ve\bigg)_{L^2}\bigg|\,,
\end{aligned}
\end{equation*}
where 
\begin{equation*}
b_1^\ve(s,u)=
-\frac{2u(\de^{\ve}/\ve)\ga(s/\ve)+u^2(\de^{\ve}/\ve)^2\ga(s/\ve)^2}{(1+u\eta^\ve(s))^2}\;;\quad
b_2^\ve(s,u)=
\frac{2u(\de^{\ve}/\ve^2)\ga'(s/\ve)}{(1+u\eta^\ve(s))^3}\;;
\end{equation*}
\begin{equation*}
W^\ve(s,u)=\frac{\de^\ve}{\ve}\bigg[
\frac{\ga(s/\ve)^2}{4}\frac{2u\ga(s/\ve)+u^2(\de^{\ve}/\ve)\ga(s/\ve)^2}{(1+u\eta^\ve(s))^2}+
\frac{u\ga''(s/\ve)}{2(1+u\eta^\ve(s))^3}-\frac{5}{4}\frac{u^2(\de^{\ve}/\ve)\ga'(s/\ve)}{(1+u\eta^\ve(s))^4}
\bigg]\,.
\end{equation*}
$b_1^\ve$, $b_2^\ve$ and $W^\ve$ are bounded functions; more precisely, there exists $\ve_0$ such that for $0< \ve< \ve_0$
\begin{equation}
\label{estb1b2W}
\|b_1^\ve\|_{L^\infty}\leqslant c\frac{\de^\ve}{\ve}\;;\quad
\|b_2^\ve\|_{L^\infty}\leqslant c\frac{\de^\ve}{\ve^2}\:;\quad
\|W^\ve\|_{L^\infty}\leqslant c\frac{\de^\ve}{\ve}\,.
\end{equation}
We use the notation
\begin{equation*}
\hat R^\ve(z):=(\hat H^\ve-z)^{-1}\,.
\end{equation*}
From the Cauchy-Schwarz inequality and estimates \eqref{stimaresolvent}  and \eqref{estb1b2W} we have
\begin{multline}
\label{start}
\big|\big(g,\big[r_{m,n}^\ve(z)-\hat r_{m,n}^\ve(z)\big]f\big)_{L^2(\RE)}\big| \leqslant
\frac{c}{|\Im z|}  \|g\|_{L^2(\RE)}
\Bigg[
\frac{\de^\ve}{\ve}\bigg\|
\pd{^2}{s^2}\hat R^\ve\lf(z+\frac{\mu_n}{\de^{\ve\,2}}\ri)f\phi_n^\ve\bigg\|_{L^2}\\
+\frac{\de^\ve}{\ve^2}
\bigg\|
\pd{}{s}\hat R^\ve\lf(z+\frac{\mu_n}{\de^{\ve\,2}}\ri)f\phi_n^\ve\bigg\|_{L^2}
+\frac{\de^\ve}{\ve^3}
\bigg\|\hat R^\ve\lf(z+\frac{\mu_n}{\de^{\ve\,2}}\ri)f\phi_n^\ve\bigg\|_{L^2}
\Bigg]\,.
\end{multline}
The estimate of the third term on the right hand side of \eqref{start} comes directly from  \eqref{stimaresolvent}
\begin{equation}
\label{estW}
\bigg\|\hat R^\ve\lf(z+\frac{\mu_n}{\de^{\ve\,2}}\ri)f\phi_n^\ve\bigg\|_{L^2}
\leqslant |\Im z|^{-1}\|f\|_{L^2(\RE)}\,.
\end{equation}

In order to estimate the other two terms we prove that there exists $\ve_0$ such that for $0< \ve< \ve_0$
\begin{equation}
\label{stimona}
\lf\|
\pd{^2}{s^2}
\Big(\hat H^\ve-\frac{\mu_n}{\de^{\ve\,2}}-z\Big)^{-1}f\phi_n^\ve
\ri\|_{  L^2 } \leqslant c \ve^{-2} \|f \|_{  L^2(\erre) }
\end{equation}
To this aim we make use of an explicit formula for the resolvent $\hat R^\ve(z)$. The proof of this formula  is postponed to appendix \ref{appendice}.

We denote by $\hat H_0^\ve$ the self-adjoint operator in $L^2$ with the same formal expression as $\hat H^\ve$  but with domain characterized by boundary conditions not depending on the $s$ variable
\begin{equation}
\label{domhatH0ve}
\DD(\hat H_0^\ve):=\bigg\{\psi\in H^2(\Omega')\,\textrm{s.t.}\,
\pd{\psi}{u} (s ,  d ) + \al \psi (s , d ) =0\,,\;
-\pd{\psi}{u} (s , - d ) + \al\psi (s , -d ) =0
 \bigg\}\,,
\end{equation}
\begin{equation}
\label{hatH0ve}
\Hd^\ve_0:=-\pd{^2}{s^2}-\frac{1}{\de^{2\,\ve}}\pd{^2}{u^2}-\frac{1}{\ve^2}\frac{\ga^2(s/\ve)}{4}\,.
\end{equation}
Moreover for all $z\in\CO\backslash\RE $ we define
\begin{equation*}
\hat R^\ve_0(z):=(\hat H^\ve_0-z)^{-1}\,.
\end{equation*}
It will be crucial in the following that the derivatives of $\hat R^\ve_0(z)$ with respect to $s$ and with respect to $u$ commute.
In particular we notice that $-\pd{^2}{s^2}-\frac{1}{\ve^2}\frac{\ga^2(s/\ve)}{4}$ and 
$-\frac{1}{\de^{2\,\ve}}\pd{^2}{u^2}$ commute with $R^\ve_0(z)$.
Let us introduce some notation and state few preliminary results. Let $L^2 (\partial \Omega') = L^2(\erre) \oplus L^2(\erre)$ and $\underline{q} \in
 L^2(\erre) \oplus L^2(\erre)$  denote a couple of functions $q_i \in  L^2(\erre)$ for $i=1,2$, $\underline q$ has  to be understood as a column vector.

For  $i=1,2$ we define two operators $G_i^\ve(z):  L^2(\erre) \to L^2$ whose integral kernels are given by
\begin{equation*}
\begin{aligned}
& G_1^\ve(z)( s,u; s') := \hat R^\ve_0(z) ( s,u; s', d) \\
& G_2^\ve(z)( s,u; s') := \hat R^\ve_0(z) ( s,u; s', -d) 
\end{aligned}
\end{equation*}
and ${\mathcal G}^\ve(z): L^2(\partial\Omega')\to L^2$ given by
\begin{equation}
\label{mathcalGve}
{\mathcal G}^\ve(z)\underline{q} := 
\begin{pmatrix}
G_1^\ve(z) & 0\\ \\
0 & G_2^\ve(z)
\end{pmatrix} 
\begin{pmatrix}
q_1 \\ \\
q_2
\end{pmatrix} =
\begin{pmatrix}
G_1^\ve(z) q_1 \\ \\ 
G_2^\ve(z) q_2
\end{pmatrix}\,.
\end{equation}
The operators $G_i^\ve(z+\mu_n/\de^{\ve\,2})$, and therefore also ${\mathcal G}^\ve(z+\mu_n/\de^{\ve\,2})$, are uniformly bounded in $\ve$ for all $n=0,1,2,\dots$
\begin{equation}
\label{stimacalG}
\Big\|\mathcal G^\ve\lf(z+\frac{\mu_n}{\de^{\ve\,2}}\ri)\underline q\Big\|_{L^2}\leqslant
c\|\underline q\|_{L^2(\partial\Omega')}\,,
\end{equation}
the proof of this statement is in appendix \ref{appendice}.

We introduce also the operators $G_{i,j}^\ve(z):  L^2(\erre) \to L^2(\erre)$ for  $i,j=1,2$ 
whose integral kernels are given by
\begin{equation*}
\begin{aligned}
& G^\ve_{1,1}(z)( s; s'): =\hat R^\ve_0(z) ( s,d; s', d) \\
& G^\ve_{1,2}(z)( s; s') :=\hat R^\ve_0(z) ( s,d; s', -d) \\
& G^\ve_{2,1}(z)( s; s') :=\hat R^\ve_0(z) ( s,-d; s', d) \\
& G^\ve_{1,1}(z)( s; s') :=\hat R^\ve_0(z) ( s,-d; s', -d) \\
\end{aligned}
\end{equation*}
and the operator $\Gamma^\ve(z) : L^2(\partial\Omega')\to L^2(\partial\Omega')$ given by
\begin{equation}
\label{Gammave}
\Gamma^\ve(z): = 
\begin{pmatrix}
G^\ve_{1,1}(z) & G^\ve_{1,2}(z) \\ \\
G^\ve_{2,1}(z) & G^\ve_{2,2}(z)
\end{pmatrix}\,.
\end{equation}
The operators $G_{i,j}^\ve(z+\mu_n/\de^{\ve\,2})$ and, consequently, $\Gamma^\ve(z+\mu_n/\de^{\ve\,2})$, are uniformly bounded in $\ve$
\begin{equation}
\label{stimaGammave}
\big\|\Gamma^\ve(z+\mu_n/\de^{\ve\,2})\underline q\big\|_{L^2(\partial\Omega')}
\leqslant
c\|\underline q\|_{L^2(\partial\Omega')}\,.
\end{equation}
For the proof of this statement one can refer to the appendix.

The resolvent $\hat R^\ve(z)$ can be written in the following form, take  $\psi \in L^2$ then 
\begin{equation}
\hat R^\ve(z) \psi = \hat R^\ve_0(z) \psi + {\mathcal G}^\ve(z)\underline{q}^\ve
\label{risolvente}
\end{equation}
where  $\underline{q}^\ve$ is defined by
\begin{equation}
\label{qve}
\underline{q}^\ve =
-\Lambda^\ve(z){\mathcal G}^{\ve\,\ast}\lf(z\ri) \psi
\qquad
\Lambda^\ve(z):=
\big(I+(\underline\al-\underline\al^\ve)\Gamma^\ve(z)\big)^{-1}(\underline\al-\underline\al^\ve)\,,
\end{equation}
where $^*$ denotes the adjoint and   $(\underline{\al} - \underline{\al}^\ve ) $ is multiplication operator
\begin{equation*}
(\underline\al-\underline\al^\ve)=
\begin{pmatrix}
(\al-\al_1^\ve) & 0 \\ \\
0 & (\al-\al_2^\ve)
\end{pmatrix} \,.
\end{equation*}
The operator $(I + (\underline{\al} - \underline{\al}^\ve ) \Gamma^\ve(z) )^{-1}$ in \eqref{qve} is well defined by its Neumann series for $\ve$ sufficiently small, see appendix \ref{appendice}. Now we go back to the proof of \eqref{stimona}. Using \eqref{risolvente} we have
\begin{equation}
\label{kumo}
\hat R^\ve \lf(z+ \frac{\mu_n}{\dde} \ri) f \phi_n^\ve = \hat R^\ve_0\lf(z+ \frac{\mu_n}{\dde}\ri) f \phi_n^\ve + 
{\mathcal G}^\ve\lf(z+ \frac{\mu_n}{\dde}\ri)\underline{q}^\ve_n
\end{equation}
with
\begin{equation*}
\underline{q}^\ve_n=-\Lambda^\ve\lf(z+ \frac{\mu_n}{\dde}\ri){\mathcal G}^{\ve\,\ast}\lf(z+ \frac{\mu_n}{\dde}\ri) f \phi_n^\ve \,.
\end{equation*}
In appendix \ref{appendice} it is proved that for $0<\ve<\ve_0$, 
\begin{equation}
\label{ultima}
\|\underline q^\ve_n\|_{L^2(\partial\Omega')}\leqslant c \frac{\de^\ve}{\ve} \|f\|_{L^2(\RE)}\,.
\end{equation}
 We first study the term of second derivative coming from the resolvent $\hat R_0^\ve(z)$. The following equality holds 
\begin{equation*}
\pd{^2}{u^2}\hat R^\ve_0(z;s,u,s',u')=
\pd{^2}{{u'}^2}\hat R^\ve_0(z;s,u,s',u')\,,
\end{equation*}
where $\hat R^\ve_0(z;s,u,s',u')$ is the integral kernel of $\hat R^\ve_0(z)$. Therefore integrating by parts
we have 
\begin{equation*}
\begin{aligned}
\bigg[\pd{^2}{u^2}\hat R^\ve_0(z)\phi_n^\ve f\bigg](s,u)=&
\int ds'\,du'\bigg[\pd{^2}{{u'}^2}\hat R^\ve_0(z;s,u,s',u')\bigg]
\phi_n^\ve(s',u')f(s')\\
=&-\big[\hat R^\ve_0(z)\la_n^\ve \phi_n^\ve f \big](s,u)+
\int ds'\hat R^\ve_0(z;s,u,s',d)(\al^\ve_1(s')-\al)\phi_n^\ve(s',d)f(s')\\
&+\int ds'\hat R^\ve_0(z;s,u,s',-d)(\al^\ve_2(s')-\al)\phi_n^\ve(s',-d)f(s')\\
=&-\big[\hat R^\ve_0(z)\la_n^\ve \phi_n^\ve f\big](s,u)
+\big[\mathcal{G}^{\ve}( z)(\underline{\al}^\ve-\underline{\al})\underline\phi_n^\ve(d)f\big](s,u)
\end{aligned}
\end{equation*}
where we used the notation $\underline\phi_n^\ve(s,d)=(\phi_n^\ve(s,d),\phi_n^\ve(s,-d))$. Then we have
\begin{equation*}
\begin{aligned}
\pd{^2}{s^2}\hat R_0^\ve(z+\mu_n/\de^{\ve\,2})  \phi_n^\ve f
=&- \phi_n^\ve f+ \lf(-\frac{1}{\de^{\ve\,2}}\pd{^2}{u^2} -\frac{1}{\ve^2}\frac{\ga^2(\cdot / \ve) }{4}
-\frac{\mu_n}{\de^{\ve\,2}}-z\ri)\hat R_0^\ve(z+\mu_n/\de^{\ve\,2})  \phi_n^\ve f \\
=&- \phi_n^\ve f+ \lf(-\frac{1}{\ve^2}\frac{\ga^2(\cdot / \ve) }{4}-z\ri)\hat R_0^\ve(z+\mu_n/\de^{\ve\,2})  \phi_n^\ve f \\
&+\hat R^\ve_0(z+\mu_n/\de^{\ve\,2})\frac{(\la_n^\ve-\mu_n)}{\de^{\ve\,2}} \phi_n^\ve f
-\mathcal{G}^{\ve}( z+\mu_n/\de^{\ve\,2})(\underline{\al}^\ve-\underline{\al})\underline\phi_n^\ve(d) f
\end{aligned}
\end{equation*}
Since $\|\ga\|_{L^\infty(\RE)}<c$ and $a>1$, then there exists $\ve_0>0$ such that, for all $0<\ve<\ve_0$, $\|\eta\|_{L^{\infty}(\erre)} \leqslant c \ve^{a-1}$ and therefore the perturbative expansion \eqref{lan2}
can be applied. This implies that
\begin{equation*}
\lf\| \frac{\la_n^\ve(\cdot) - \mu_n}{\dde} \ri\|_{L^{\infty}(\erre) } \leqslant c \frac{1}{\ve^2}\,.
\end{equation*}
Moreover from $\|(\underline{\al}^\ve-\underline{\al})\|_{L^\infty(\partial\Omega')}\leqslant c\de^\ve/\ve$ and estimate \eqref{stimacalG} we have that  there exists $\ve_0$ such that for all $0<\ve<\ve_0$
\begin{equation}
\lf\|
\pd{^2}{s^2} \hat R^\ve_0\lf(z+ \frac{\mu_n}{\dde}\ri) f \phi_n^\ve
\ri\|_{L^2 } 
\leqslant
c\Big(1+\frac{1}{\ve^2}+\frac{\de^{\ve}}{\ve}\Big)\|f\|_{L^2(\RE)}
\leqslant c \frac{1}{\ve^2} \|f \|_{L^2(\erre) }\,.
\label{tappo}
\end{equation}
Now we estimate the second term coming from the r.h.s. of \eqref{kumo}. We can write it in the
following equivalent way
\begin{equation*}
{\mathcal G}^\ve\lf(z+ \frac{\mu_n}{\dde}\ri)\underline{q}^\ve_n = -
{\mathcal G}^\ve\lf(z+ \frac{\mu_n}{\dde}\ri)
\Lambda^\ve\lf(z+ \frac{\mu_n}{\dde}\ri) 
{\mathcal G}^{\ve\,\ast}\lf(z+ \frac{\mu_n}{\dde}\ri) f \phi_n^\ve\,.
\end{equation*}
We introduce the notation
\begin{equation*}
\hat h^\ve:=-\pd{^2}{s^2}-\frac{\gamma(s/\ve)^2}{4\ve^2}\,.
\end{equation*}
Let us notice that $\hat h^\ve$ commutes with $\hat R^\ve_0(z)$ and $\Gamma^\ve(z)$. Using the commutation property of $\hat h^\ve$ we have
\begin{equation}
\label{quasi}
\begin{aligned}
\pd{^2}{s^2} {\mathcal G}^\ve\lf(z+ \frac{\mu_n}{\dde}\ri)\underline{q}^\ve_n =&
-\lf( \frac{1}{\ve^2}\frac{\ga^2(\cdot / \ve) }{4} +z\ri) {\mathcal G}^\ve\lf(z+ \frac{\mu_n}{\dde}\ri)\underline{q}^\ve_n \\
&-{\mathcal G}^\ve\lf(z+ \frac{\mu_n}{\dde}\ri) 
\lf( \hat h^\ve -z\ri)\Lambda^\ve\lf(z+ \frac{\mu_n}{\dde}\ri)  
{\mathcal G}^{\ve\,\ast}\lf(z+ \frac{\mu_n}{\dde}\ri) f \phi_n^\ve\,.
\end{aligned}
\end{equation}
By estimate \eqref{stimacalG} and \eqref{ultima} one can see that the first term on the r.h.s. of \eqref{quasi} is bounded by  $c\de^\ve/\ve^{3} \| f\|_{L^2(\erre)}$. The second term of \eqref{quasi} can be written in the following way
\begin{equation*}
\begin{aligned}
&{\mathcal G}^\ve\lf(z+ \frac{\mu_n}{\dde}\ri) 
\lf(\hat h^\ve -z\ri)\Lambda^\ve\lf(z+ \frac{\mu_n}{\dde}\ri)  
{\mathcal G}^{\ve\,\ast}\lf(z+ \frac{\mu_n}{\dde}\ri) f \phi_n^\ve\\
=&{\mathcal G}^\ve\lf(z+ \frac{\mu_n}{\dde}\ri) 
\lf(\hat h^\ve-z\ri) 
\Lambda^\ve\lf(z+ \frac{\mu_n}{\dde}\ri) 
\lf(\hat h^\ve -z\ri)^{-1} 
\lf( \hat h^\ve -z\ri) 
{\mathcal G}^{\ve\,\ast}\lf(z+ \frac{\mu_n}{\dde}\ri) f \phi_n^\ve\,. 
\end{aligned}
\end{equation*}
By an argument similar to the one used in the proof of estimate  \eqref{tappo} one can prove that
\begin{equation}
\lf\|
\lf(\hat h^\ve-z\ri) 
{\mathcal G}^{\ve\,\ast}\lf(z+ \frac{\mu_n}{\dde}\ri) f \phi_n^\ve 
\ri\|_{L^2(\partial\Omega')} \leqslant c \frac{1}{\ve^2} \| f \|_{L^2(\RE)}\,.
\label{punto2}
\end{equation}
Moreover the following identity holds
\begin{equation*}
\begin{aligned}
&\lf(\hat h^\ve-z\ri)\Lambda^\ve  
\lf(z+ \frac{\mu_n}{\dde}\ri) 
 \lf(\hat h^\ve-z\ri)^{-1} \\
  =&\lf[ I + \lf(\hat h^\ve -z\ri) 
  (\underline{\al} - \underline{\al}^\ve )  
\lf( \hat h^\ve-z\ri)^{-1}
\Gamma^\ve\lf(z+ \frac{\mu_n}{\dde}\ri)   
\ri]^{-1}
\lf(\hat h^\ve -z\ri) 
  (\underline{\al} - \underline{\al}^\ve )  
\lf(\hat h^\ve -z\ri)^{-1} \,.
\end{aligned}
\end{equation*}
Using the Leibniz rule, we see that
\begin{equation*}
\lf\|
\lf(\hat h^\ve-z\ri) 
  (\underline{\al} - \underline{\al}^\ve )  
\lf( \hat h^\ve -z\ri)^{-1}
\ri\|_{\BB(L^2(\partial\Omega') , L^2(\partial\Omega')) } 
\leqslant c\Big(
\frac{\de^\ve}{\ve}+\frac{\de^\ve}{\ve^{2}}+\frac{\de^\ve}{\ve^3}\Big) \leqslant c \frac{\de^\ve}{\ve^3}\,.
\end{equation*}
Then there  exists $\ve_0$ such that
for $0<\ve < \ve_0$ 
\begin{equation}
\lf\|\lf(\hat h^\ve-z\ri)\Lambda^\ve  
\lf(z+ \frac{\mu_n}{\dde}\ri) 
 \lf(\hat h^\ve-z\ri)^{-1}
\ri\|_{\BB(L^2(\partial\Omega') , L^2(\partial\Omega') )}\leqslant c \frac{\de^\ve}{\ve^3}\,.
\label{punto1}
\end{equation}
Therefore it follows from estimates \eqref{stimacalG} and 
\eqref{punto2} and \eqref{punto1} that
\begin{equation}
\label{sonno}
\lf\|\pd{^2}{s^2}
{\mathcal G}^\ve\lf(z+ \frac{\mu_n}{\dde}\ri)\underline{q}^\ve
\ri\|_{L^2} 
\leqslant c \frac{\de^\ve}{\ve^5}\|f \|_{  L^2(\erre) }\,.
\end{equation}
Then for $a>3$,  \eqref{stimona} follows from \eqref{sonno} and \eqref{tappo}.
Interpolating \eqref{stimona} and using the $L^2$ boundedness of
$(\hat H^\ve-\frac{\mu_n}{\de^{\ve\,2}}-z)^{-1}f\phi_n^\ve $ we immediately obtain
\begin{equation}
\lf\|
\pd{}{s}
\Big(\hat H^\ve-\frac{\mu_n}{\de^{\ve\,2}}-z\Big)^{-1}f\phi_n^\ve
\ri\|_{  L^2 } \leqslant c \frac{1}{\ve} \|f \|_{  L^2(\erre) }\,.
\label{stimina}
\end{equation}

The proof of \eqref{goal} comes from \eqref{start}, \eqref{estW}, \eqref{stimona} and \eqref{stimina}.
\end{proof}
Let us consider the family of self-adjoint operators $\hh^\ve_n:H^2(\RE)\subset L^2(\RE)\to L^2(\RE)$, $n=0,1,2,\dots$ defined in \eqref{hven}. The following lemma concludes the first step in the proof of theorem \ref{mainth}; it shows that
in the limit only the diagonal elements of the reduced resolvent survive and that the leading  term in the reduced Hamiltonian is $h_n^\ve$.

\begin{lemma}
\label{lemma2}
Let $\ga\in C_0^\infty(\RE)$ and  $a>3$ then for all $m,n=0,1,2,\dots$ and for all $z\in\CO\backslash\RE$,
\begin{equation*}
\ulim_{\ve\to0}
\big(\rbar_{m,n}^\ve(z)-\de_{m,n}(\hh^\ve_n-z)^{-1}\big)=0\,.
\end{equation*}
\end{lemma}
\begin{proof}
It is sufficient to prove that for all $m,n=0,1,2,\dots$, $f,g\in C_0^\infty(\RE)$ and $z\in\CO\backslash\RE$, there exists $\ve_0>0$ such that, for all $0<\ve<\ve_0$,
\begin{equation*}
\big|\big(g,(\rbar_{m,n}^\ve(z)-\de_{m,n}(\hh^\ve_n-z)^{-1})f\big)_{L^2(\RE)}\big|
\leqslant c \ve^{a-3} \|g\|_{L^2(\RE)}\|f\|_{L^2(\RE)}\,.
\end{equation*}
It is convenient to introduce an intermediate Hamiltonian $\hat{h}_n^\ve$ which is the compression of
$\hat H^\ve$ to the subspace generated by $\phi_n^\ve$. We define
\begin{equation*}
\hat h_n^\ve = -\frac{d^2}{ds^2}-\frac{1}{\ve^2}\frac{\ga^2(s/\ve)}{4}+\frac{\la_n^\ve(s)-\mu_n}{\dde}\,.
\end{equation*}
Since
\begin{equation}
\label{ineq1}
\begin{aligned}
&\big|\big(g,(\rbar_{m,n}^\ve(z)-\de_{m,n}(\hh^\ve_n-z)^{-1})f\big)_{L^2(\RE)}\big|\\
\leqslant& 
\big|\big(g,(\rbar_{m,n}^\ve(z)-\de_{m,n}(\hat h_n^\ve-z)^{-1} )f\big)_{L^2(\RE)}\big|+
\de_{m,n}\big|\big(g,(\hat h_n^\ve - z)^{-1}-(\hh^\ve_n-z)^{-1})f\big)_{L^2(\RE)}\big|\,,
\end{aligned}
\end{equation}
it is sufficient to estimate separately the two terms on the right hand side of \eqref{ineq1}.

We notice that
\begin{equation*}
\phi_n(\underline{\al}^\ve(s))=\xi_n+\varphi^\ve_n(s)
\end{equation*}
where $\xi_n$ are the eigenfunctions of the symmetric Robin Laplacian in $(-d,d)$, see \eqref{xin1} - \eqref{xin2}
and $\varphi^\ve_n(s)$ by elementary calculus is given by the following line integral in $\erre^2$
\begin{equation}
\label{344}
\varphi^\ve_n(s)= \int_{[ \underline{\al} , \underline{\al}^\ve(s)]  } \nabla_{\underline\nu}\phi_n(\underline{\nu}) 
\cdot d\underline{\nu}
\end{equation}
where with a small abuse of notation we have denoted the segment in $\erre^2$ between $\underline{\al}$ and 
$\underline{\al}^\ve(s)$ by $[ \underline{\al} , \underline{\al}^\ve(s)]$. Notice that 
\begin{equation}
\|\al_j^\ve -\al \|_{L^{\infty}(\erre)} \leqslant c \de^\ve / \ve
\qquad \qquad
j=1,2
\,,
\label{spect}
\end{equation}
which implies $\de\!{\al} 
\equiv \sup_{s\in\RE}|  \underline{\al} - \underline{\al}^\ve(s) | \leqslant c \de^\ve / \ve$.
The eigenstates $\phi_n(\underline \nu)$ are $L^2((-d,d))$-valued smooth functions of $\underline \nu\in\RE^2$,
in particular $\| \phi_n(\underline \nu  ) \|_{L^2((-d,d)) }=1$ and for all $\underline\al\in \RE^2$ and $0<\ve<\ve_0$ there exists an open ball $B(\underline\al,R^\ve)\subset\RE^2$ with center in $\underline \al$ and radius $R^\ve<c\de^\ve/\ve$ such that 
\begin{equation}
\label{cercare}
\sup_{\underline \nu \in  B(\underline\al,R^\ve)} \|\partial_{\nu_i} \phi_n(\underline \nu  ) \|_{L^2((-d,d)) } \leqslant c\;,
\qquad  
\sup_{\underline \nu \in B(\underline\al,R^\ve)} \|\partial_{\nu_i}\partial_{\nu_j} \phi_n(\underline \nu  ) \|_{L^2((-d,d)) } \leqslant c
\end{equation}
for $i,j=1,2$.

Let $T$ be a given bounded operator in $L^2$ and let $T(s,u;s',u')$ be its integral kernel. 
For fixed $f \in L^2(\erre)$ we introduce 
\begin{align*}
\zeta_1^f (s,u) &=  \int_{\Omega'} ds' du' \,T(s,u;s',u')\phi_n(\underline{\al}^\ve(s'),u') f(s') \\
\zeta_2^f (s,u) &=  \int_{\Omega'} ds' du' \,T(s,u;s',u')\phi_n(\underline{\al}^\ve(s),u') f(s') 
\end{align*}
and we want to prove
\begin{equation*}
\| \zeta_1^f - \zeta_2^f \|_{L^2} \leqslant c \de^\ve / \ve \| f\|_{L^2(\erre)}\,.
\end{equation*}
It is convenient to introduce also
\begin{equation*}
\zeta_0^f (s,u) =  \int_{\Omega'} ds' du' \,T(s,u;s',u')\xi_n(u') f(s') \\
\end{equation*}
and separately estimate $\| \zeta_1^f - \zeta_0^f \|_{L^2}$ and $\| \zeta_2^f - \zeta_0^f \|_{L^2}$.
The estimate of $\| \zeta_1^f - \zeta_0^f \|_{L^2}$ trivially follows from the boundedness of $T$ and \eqref{cercare}:
\begin{equation*}
\| \zeta_1^f - \zeta_0^f \|_{L^2}\leqslant 
\|T\|_{\BB(L^2,L^2)} \| (\phi_n -\xi_n) f\|_{L^2} \leqslant c  \de^\ve / \ve \| f\|_{L^2(\erre)}\,.
\end{equation*}
The estimate of $\| \zeta_2^f - \zeta_0^f \|_{L^2}$ requires a more careful analysis: interchanging integrals,
we have
\begin{align}
\label{tac}
\lf(\zeta_2^f - \zeta_0^f \ri) (s,u) & = 
\int_{\al}^{\al_1^\ve (s)} d\nu 
\int_{\Omega'} ds' du' \,T(s,u;s',u') \partial_\nu \phi_n (\nu,\al,u') f(s') + \\ \nonumber
& +\int_{\al}^{\al_2^\ve (s)} d\nu 
\int_{\Omega'} ds' du' \,T(s,u;s',u') \partial_\nu \phi_n (\al,\nu,u') f(s') + \\ \nonumber
& +\int_{\al}^{\al_1^\ve (s)} d\nu_1 
\int_{\al}^{\al_2^\ve (s)} d\nu_2
\int_{\Omega'} ds' du' \,T(s,u;s',u') \partial_{\nu_1} \partial_{\nu_2}\phi_n (\nu_1,\nu_2,u') f(s')\,.
\end{align}
Using \eqref{spect}, \eqref{cercare} and \eqref{tac} it is straightforward to prove
\begin{equation*}
\| \zeta_2^f - \zeta_0^f \|_{L^2} \leqslant 
c \de^\ve/ \ve \|T \|_{\BB(L^2,L^2)} \|f\|_{L^2(\erre)}\,.
\end{equation*}

Let us come back to the proof of \eqref{ineq1}. We notice the following useful identity 
\begin{equation*}
\delta_{m,n}\big(g,(\hat h_n^\ve-z)^{-1} f\big)_{L^2(\RE)} =
\lf( g \phi_m , \zeta_2^f \ri)_{L^2}
\end{equation*}
with $T = \lf(\Hd^\ve-\mu_n/\dde-z\ri)^{-1}$. Therefore we have
\begin{equation*}
\big(g,(\rbar_{m,n}^\ve(z)-\de_{m,n}(\hat h_n^\ve-z)^{-1})f\big)_{L^2(\RE)}=
\lf( g \phi_m, (\zeta_1^f - \zeta_2^f) \ri)_{L^2}
\end{equation*}
and from the Cauchy-Schwarz inequality we get
\begin{equation}
\label{stima1}
\begin{aligned}
&\big|\big(g,(\rbar_{m,n}^\ve(z)-\de_{m,n}(\hat h_n^\ve-z)^{-1})f\big)_{L^2(\RE)}\big|=
\lf|\lf( g \phi_m, (\zeta_1^f - \zeta_2^f) \ri)_{L^2}\ri|  \\
\leqslant&\| g\|_{L^2(\erre)}\| \zeta_1^f - \zeta_2^f \| \leqslant c \de^\ve / \ve 
| \Im z|^{-1}\| g\|_{L^2(\erre)} \| f\|_{L^2(\erre)}
\end{aligned}
\end{equation}
and this concludes the  estimate of the first term at the right hand side of \eqref{ineq1}.

Let us consider now the second term at the right hand side of equation \eqref{ineq1}. From the first resolvent identity we get
\begin{equation*}
\big|\big(g,\big((\hat h_n^\ve-z)^{-1}-(\hh^\ve_n-z)^{-1}\big)f\big)_{L^2(\RE)}\big|=
\bigg|
\bigg(g,(\hat h_n^\ve-z)^{-1}\bigg[\frac{\la_n^\ve(\cdot)-\mu_n-\la_n^{(2)}\eta^{\ve}(\cdot)^2}{\dde}\bigg]
(\hh^\ve_n-z)^{-1}f\bigg)_{L^2(\RE)}\bigg|\,.
\end{equation*}
By the Cauchy-Schwarz inequality we have
\begin{equation}
\label{stima2}
\big|\big(g,\big((\hat h_n^\ve-z)^{-1}-(\hh^\ve_n-z)^{-1}\big)f\big)_{L^2(\RE)}\big|
\leqslant
|\Im z |^{-2}
\bigg\|\frac{\la_n^\ve(\cdot)-\mu_n-\la_n^{(2)}\eta^{\ve}(\cdot)^2}{\dde}\bigg\|_{L^\infty(\RE)}
\|f\|_{L^2(\RE)}\|g\|_{L^2(\RE)}\,.
\end{equation}
Since $\|\ga\|_{L^\infty(\RE)}<c$ and $a>1$, then there exists $\ve_0>0$ such that, for all $0<\ve<\ve_0$, $\|\eta\|_{L^{\infty}(\erre)} \leqslant c \ve^{a-1}$ and therefore the perturbative expansion \eqref{lan2}
can be applied. This implies that
\begin{equation*}
\|\la_n^\ve(\cdot)-\mu_n-\la_n^{(2)}\eta^{\ve}(\cdot)^2\|_{L^\infty(\RE)}\leqslant c\ve^{3(a-1)}\,.
\end{equation*}
Then we have
\begin{equation*}
\bigg\|
\frac{\la_n^\ve(\cdot)-\mu_n-\la_n^{(2)}\eta^{\ve}(\cdot)^2}{\dde}\bigg\|_{L^\infty(\RE)}
\leqslant
c\ve^{a-3}\,.
\end{equation*} 
From the last estimate and from equation  \eqref{stima2} we get
\begin{equation*}
\big|\big(g,\big((\hat h_n^\ve-z)^{-1}-(\hh^\ve_n-z)^{-1}\big)f\big)_{L^2(\RE)}\big|
\leqslant
\ve^{a-3} c |\Im(z)|^{-2}
\|f\|_{L^2(\RE)}\|g\|_{L^2(\RE)}\,,
\end{equation*}
that together with equation \eqref{stima1} concludes the proof of the lemma.
\end{proof}

{\bf Proof of theorem \ref{mainth}}

The proof of theorem \ref{mainth} follows directly from lemma \ref{lemma1}, lemma \ref{lemma2}
and proposition \eqref{prop1}.  \hfill $\square$ \newline

\section{Small deformations of the curvature \label{sec5}}
\setcounter{equation}{0}

In this section we give a generalization of the previous results: we show that deformations of the order of $\ve$ of the angle $\theta$ can lead to a more general coupling in the vertex.

Let us consider the following scaling for the curvature 
\begin{equation}
\label{scalinggamma2}
\ga(s) \longrightarrow \,\,\frac{1}{\ve}\widetilde\ga^\ve\lf(\frac{s}{\ve}\ri)=\frac{\sqrt{1+2 \ve b}}{\ve} \ga\lf( \frac{s}{\ve} \ri) 
\qquad\ve>0\,,
\end{equation}
where $b$ is a real constant. With this scaling the angle $\theta$ between the straight parts of the curve  $l_1$ and $l_2$ is
\begin{equation*}
\theta^\ve=\int_\RE\widetilde\gamma^\ve(s)ds=\sqrt{1+2 \ve b}\,\theta=(1+\ve b)\theta+\OO(\ve^2)\,,
\end{equation*}
then the  scaling \eqref{scalinggamma2} can be interpreted as a deformation, of order $\ve$, of the geometric parameter $\theta$.

Consider the family of one dimensional Hamiltonians with scaled potential of the form
\begin{equation}
\label{barhveb}
\widetilde h^\ve:=-\frac{d^2}{ds^2}+\frac{1+\ve b}{\ve^2}v(s/\ve)\,.
\end{equation}
If  exists a zero energy resonance for the Hamiltonian $h=-\frac{d^2}{ds^2}+v$  one can define two constants  $c_-$ and $c_+$ as it was done in \eqref{c1c2}. The  family of Hamiltonians 
\begin{equation}
\label{domhrb}
\begin{aligned}
{\mathscr D}(\widetilde h_r) := \{ f\in H^2(\erre \setminus 0 ) \,\, s.t. \,\,&
c_- f(0^+ ) = c_+ f(0^- )\, ,\\
&c_+ f' (0^+ ) -c_- f' (0^- )=\hat b(c_- f(0^- ) + c_+ f(0^+)) \}
\end{aligned}
\end{equation}
where 
\begin{equation*}
\hat b:=b\int_\RE v(s)\big(f_r(s)\big)^2 ds
\end{equation*}
 and 
\begin{equation}
\label{hrb}
\widetilde h_r f := - \frac{d^2 f}{ds^2} \qquad s\neq 0\,,
\end{equation}
 individuates a family of self-adjoint perturbations of the Laplacian in dimension one (see, e.g., \cite{ABD95}).
 
The following proposition generalizes the result stated in proposition \ref{prop1}, the proof can be read in \cite{CE07}, theorem 3.1. (see also \cite{ACF07}).
\begin{proposition}
\label{prop2}
Take $\widetilde h^\ve$ and $h$  defined as above and assume that $v$ satisfies conditions \eqref{vass}. Then two cases can occur:
\begin{enumerate}
\item There does not exist a zero energy resonance for the Hamiltonian $h$, then
\begin{equation*}
\ulim_{\ve\to0} ( \widetilde h^\ve - z )^{-1} = ( h_0 -z )^{-1} \qquad
 z \in \CO \setminus \erre\,.
\end{equation*}
\item There exists a zero energy resonance for the Hamiltonian $h$, then
\begin{equation*}
\ulim_{\ve\to0} (\widetilde h^\ve - z )^{-1} = (\widetilde h_r -z )^{-1} \qquad
z \in \CO \setminus \erre\,.
\end{equation*}
\end{enumerate}
\end{proposition}

Let us denote by $\widetilde \Omega^\ve$ the family of domains obtained from $\Omega$ by scaling $\gamma$ and $d$ as stated in \eqref{scalinggamma2} and \eqref{scalingd}. Following what was done in section \ref{sec1} one can define the Robin Laplacian on the family of domains $\widetilde \Omega^\ve$, the operator $-\Delta_{\widetilde \Omega^\ve}^R$ is unitarily equivalent to the operator $\widetilde H^\ve$ obtained via the substitution  $\gamma\to\widetilde\gamma^\ve=\sqrt{1+2\ve b} \,\gamma$ in \eqref{Hve}.
\begin{equation*}
\widetilde H^{\ve}:=-
\pd{}{s} \frac{1}{(1+u\widetilde \eta^\ve(s))^2}\pd{}{s}-
 \frac{1}{\de^{\ve\, 2}} \pd{^2}{u^2} +\frac{1}{\ve^2} \widetilde V^\ve(s,u)\,,
\end{equation*}
\n with
\begin{equation*}
\widetilde V^\ve(s,u):=-\frac{\widetilde \gamma^\ve(s/\ve)^2}{4(1+u\widetilde \eta^\ve(s))^2}
+\frac{\de^\ve/\ve \, u\widetilde {\gamma^\ve}''(s/\ve)}{2(1+u\widetilde \eta^\ve(s))^3}
-\frac{5}{4}\frac{(\de^\ve/\ve)^2 u^2\widetilde {\gamma^\ve}'(s/\ve)^2}{(1+u \widetilde\eta^\ve(s))^4}
\end{equation*}
and with domain
\begin{equation*}
\DD(\widetilde H^\ve):=\bigg\{\psi\in H^2(\Omega')\,s.t.\,
 \pd{\psi}{u} (s , d ) + \widetilde \al^{\ve}_1 (s) \psi (s ,d ) =0\;,
-\pd{\psi}{u} (s ,- d ) + \widetilde \al^{\ve}_2 (s) \psi (s , -d ) =0
\bigg\}\,,
\end{equation*}
here  $\widetilde\al_1^\ve (s)$ and $\widetilde\al_2^\ve (s)$ are given by
\begin{equation}
\widetilde\al_1^\ve (s):=\al-\frac{\widetilde\eta^\ve (s)}{2(1+d \widetilde\eta^\ve(s))}\;,\quad
\widetilde\al_2^\ve (s):=\al+\frac{\widetilde\eta^\ve (s)}{2(1- d \widetilde\eta^\ve(s))}
\end{equation}
and we have introduced $\widetilde\eta^\ve(s): = \de^\ve / \ve \, \widetilde\ga^\ve(s/\ve)$. 

It is easy to understand  that a slightly different version of lemma \ref{lemma1} and lemma \ref{lemma2} holds  for the Hamiltonian $\widetilde H^\ve$.

With the scaling \eqref{scalinggamma2}  the  main contribution to the longitudinal part of the Hamiltonian, once the dynamics has been reduced to the $n$-th transverse mode, is given by
\begin{equation*}
\widetilde h^\ve_n:=-\frac{d^2}{ds^2}+\bt_n\frac{1+\ve b}{\ve^2}\ga^2(s/\ve)\,,
\end{equation*}
that via  proposition \ref{prop2} leads to a slightly different version of theorem \ref{mainth}. 

For all $z\in{\CO\backslash\RE}$ and for all $m,n=0,1,2,\dots$, let us denote by  $\widetilde r_{m,n}^\ve(z)$  the reduced resolvent of $\widetilde H^\ve$, i.e., the operator in $\BB\big(L^2(\RE),L^2(\RE)\big)$ defined via its integral kernel by
\begin{equation*}
\widetilde r_{m,n}^\ve(z;s,s'):=
\int_{-d}^{d}\int_{-d}^{d}
\widetilde\phi_m^\ve(s,u)\,\Big(\widetilde H^\ve-\frac{\mu_n}{\dde}-z\Big)^{-1}(s,u;s',u')\,\widetilde\phi_n^\ve(s',u')
du\,du'\,.
\end{equation*}

\begin{theorem}
Assume that $\Ga$ has no self-intersections and that $\ga\in C_0^\infty(\RE)$, moreover take $a>3 $ then for all $n,m=0,1,2,\dots$ two cases can occur:
\begin{enumerate}
\item For all $n$ such that  there does not exist a zero energy resonance for  $h_n$ we have
\begin{equation*}
\ulim_{\ve\to0}\widetilde r_{m,n}^\ve(z)= \delta_{m,n}( h_0 -z )^{-1} \qquad
 z \in \CO \setminus \erre\,,\;m,n=0,1,2,\dots\,.
\end{equation*}
\item For all $n$ such that there exists  a zero energy resonance, $f_{r,n}$, for  $h_n$ we have
\begin{equation*}
\ulim_{\ve\to0}\widetilde r_{m,n}^\ve(z)=\delta_{m,n} (\widetilde h_{r,n} -z )^{-1} \qquad
z \in \CO \setminus \erre\,,\;m,n=0,1,2,\dots\,.
\end{equation*}
where $\widetilde h_{r,n}$ is  defined according to  \eqref{domhrb} and  \eqref{hrb}.
\end{enumerate}
\end{theorem}

\section{Conclusions and remarks}
\setcounter{equation}{0}

We have studied the convergence of the Robin Laplacian on a waveguide to an operator on a graph made up of two edges and one vertex. In this setting we were able to give detailed results on the convergence in norm resolvent sense. Our analysis takes into account the projection on each transverse mode and distinguishes generic and non-generic cases leading respectively to decoupling and non-decoupling gluing conditions in the vertex. Non-generic cases are related to the existence of zero energy resonances for the leading term of the effective Hamiltonian. The relation between gluing conditions in the vertex and the resonance is rigorously stated in equations  \eqref{c1c2} - \eqref{hr}. 

The existence of a zero energy  resonance is an exceptional event and in general it is destroyed by slight deformations of the potential. As it is clearly shown in figure \ref{fig2}, apart for  few special cases, the  coefficients $\bt_{n}$ change as $n$ changes. For these reasons even in the simpler example of a waveguide in most of the cases the operator on the corresponding graph is defined by conditions in the vertex of decoupling type and in general,  for fixed $\al$, one can expect coupling at most in one transverse mode. This is in agreement with previous results derived in \cite{ACF07}, \cite{CE07}, \cite{MV07pp} and \cite{Gri07pp}.

It is interesting to note that in this simple model all the interplay between geometry and boundary
conditions on the initial domain, is reduced to the value of $\beta_n$. In particular a positive sign of
$\beta_n$, may rule out the possibility of having a zero energy resonance giving
decoupling conditions in the limit. We conjecture that much more complicate 
geometries than the vertex region of a strip with constant
width may give different threshold singularities opening up for the possibility of much more general non decoupling
conditions.

We discuss with more detail our result for some suitable choices of the parameter $\al$.

Formally the limit $\al\to\pm\infty$ gives  Dirichlet conditions on the boundary of the waveguide. This case was discussed in two former works \cite{ACF07} and \cite{CE07}, the results stated there are included in our paper. This can be seen by studying the asymptotic behavior of equation \eqref{btn} and in particular by noticing that $\bt_n(\pm\infty)=-1/4$, see also figure \ref{fig2}. The anomalous behavior of $\bt_0(\al)$ and $\bt_1(\al)$ 
for $\al \to -\infty$  is related to the existence of  two negative eigenvalues for $\hhd_\al$ when $\al<-1$, see the plot of $\mu_0(\al)$ and $\mu_1(\al)$ in figure \ref{fig1}. It is important to notice that Dirichlet boundary conditions on $\partial\Omega$ are preserved by the unitary map \eqref{unitarymap}. As a consequence the analysis of the transverse part of the Hamiltonian is particularly simple in this case: transverse modes and eigenvalues do not depend on $s$ and the Hamiltonian  on the graph does not depend on $n$.

The case $\al=0$ reproduces Neumann boundary conditions. As expected even in our model the Neumann waveguide, on the ground state, is approximated by free conditions in the vertex and  this behavior is summarized by $\bt_0(0)=0$. The  reduced Hamiltonian relative to the lowest transverse mode in the Neumann case is the free Laplacian  on the line and  the  zero energy resonance is the constant function. In this case the constants $c_-$ and $c_+$ coincide and the limit Hamiltonian is defined by free conditions in the vertex.

Since in the Neumann case $\bt_n(0)=3/4$ for $n=1,2,\dots$, the Hamiltonian on the graph obtained by projecting onto the excited transverse modes has decoupling conditions in the vertex.

We want to stress an important  difference between our approach and the one generically used in other works. In our setting  the domain $\Omega^\ve$  is defined via  two characteristic lengths: the width of the waveguide $\de^\ve$ and the range $\ve$ on which the curvature varies. Since we assume that for small $\ve$ we have $\de^\ve\ll\ve$, we can make use of the adiabatic separation of the dynamics that leads in a natural way toward  an analysis of the problem in two steps. The existence of two different scales of length makes it difficult to compare our model with the one generically used in the works cited in  the introduction in which the scaling of the network in a neighborhood of the vertex is assumed to be isotropic.

The analysis of the case $\de^\ve=\ve$ would be of great interest because it reproduces the scaling of  \cite{Gri07pp} and \cite{MV07pp}. The failing of the adiabatic approach makes this case much more complicated and we guess that in this setting the uniform resolvent convergence could be too demanding.

Even in more complicated settings, such as graphs with three or more edges, it  seems reasonable to expect that deformations of order $\ve$ of networks associated to non decoupling conditions in the vertex could lead to more general gluing  conditions. This idea is suggested from the result stated in section \ref{sec5} and was already envisaged in \cite{CE07}.\\

\vspace{0.3cm}

{\bf Acknowledgments} The authors are grateful to Sergio Albeverio,
Pavel Exner and David Krej\v{c}i\v{r}\'ik for the  enlightening hints and useful comments. This work was supported by the Doppler Institute grant (LC06002).\\

\appendix
\section{Proof of formula \eqref{risolvente} \label{appendice}}
\setcounter{equation}{0}

We mark that formula \eqref{risolvente} can be proved by making use of a very general technique developed by A. Posilicano in \cite{Pos01} and \cite{Pos07pp}. For convenience of the reader we give a more direct proof of the formula in this appendix together with the proof of several technical estimates used in lemma \ref{lemma1}.  

Let us start with the proof of formula \eqref{risolvente}. We denote by $\hat{Q}_0^\ve$ the quadratic form associated to the operator $\hat H^\ve_0$ defined in \eqref{domhatH0ve} - \eqref{hatH0ve}, it  is given by
\begin{equation*}
\begin{aligned}
\hat{Q}_0^\ve [\varphi,\psi]=&
\int_{\Omega'}\bigg( \overline{\pd{\varphi}{s}}\pd{\psi}{s}+
\frac{1}{\dde}\overline{\pd{\varphi}{u}}\pd{\psi}{u}-
\frac{1}{\ve^2}\frac{\ga^2}{4}\,\overline{\varphi}\,\psi\bigg) ds\,du \\
&+\al\int_\RE(\overline{\varphi}(s,d)\psi(s,d)+\overline{\varphi}(s,-d)\psi(s,-d))ds\,.
\end{aligned} 
\end{equation*}
We recall also that the quadratic form $\hat{Q}^\ve$ associated to $\hat H^\ve$ defined in \eqref{domhami} - \eqref{hami} is
\begin{equation}
\label{formetta2}
\hat{Q}^\ve [\varphi,\psi]=\hat{Q}_0^\ve [\varphi,\psi]
+\int_\RE
\lf[
(\al_1^\ve(s)-\al)\overline{\varphi}(s,d)\psi(s,d)+(\al_2^\ve(s)-\al)\overline{\varphi}(s,-d)\psi(s,-d)\ri]ds\,.
\end{equation}
Let us prove that for all $\psi \in L^2$
\begin{equation*}
\hat R^\ve(z) \psi = \hat R^\ve_0(z) \psi + {\mathcal G}^\ve(z)\underline{q}^\ve
\end{equation*}
where ${\mathcal G}^\ve(z)$ was defined in  \eqref{mathcalGve} and  $\underline{q}^\ve$ is a solution of equation
\begin{equation}
\label{charge}
\underline{q}^\ve + (\underline{\al} - \underline{\al}^\ve ) 
\lf( \Gamma^\ve(z) \underline{q}^\ve + {\mathcal G}^{\ve\,\ast}(z) \psi \ri) = 0\,.
\end{equation}
For $\psi \in L^2$, we define 
\begin{equation*}
T^\ve(z) \psi = \hat R^\ve_0(z) \psi + {\mathcal G}^\ve(z)\underline{q}^\ve
\end{equation*}
and we look for conditions on $\underline{q}^\ve$ such that 
\begin{equation}
\label{formina}
\hat{Q}^\ve [\varphi, T^\ve(z)\psi] -z ( \varphi , T^\ve(z)\psi ) = ( \varphi , \psi )
\end{equation}
for any $\varphi \in C^{\infty}_0 (\erre^2 ) $. If \eqref{formina} holds then $T^\ve(z) = \hat R^\ve(z)$. If we start from \eqref{formetta2}, integrate by parts and notice that $\hat R^\ve_0(z) \psi
\in \DD ( \hat H_0^\ve ) $, we have
\begin{equation*}
\hat{Q}^\ve [\varphi, T^\ve(z)\psi] -z ( \varphi , T^\ve(z)\psi ) = ( \varphi , \psi )+
\lf( \underline{\varphi}  , 
\lf[
\underline{q}^\ve + (\underline{\al} - \underline{\al}^\ve ) 
\lf( \Gamma^\ve(z) \underline{q}^\ve + {\mathcal G}^{\ve\,\ast}(z) \psi \ri) 
\ri]
\ri)_{L^2(\partial\Omega')}
\end{equation*}
with $\underline{\varphi}(s) =(\varphi(s,d),\varphi(s,-d))$, which implies \eqref{charge}. Now we have to prove that equation 
\eqref{charge} has a unique solution for all $\psi\in L^2$. To this aim we need to prove that $\Gamma^\ve(z)$ and ${\mathcal G}^{\ve}(z)$ are bounded. Let us prove estimates \eqref{stimacalG} and \eqref{stimaGammave} that are required in the proof of lemma \ref{lemma1}, then the boundedness of $\Gamma^\ve(z)$ and ${\mathcal G}^{\ve}(z)$ will be an obvious consequence.

To prove \eqref{stimacalG}  we need to prove that that $G_i^\ve(z+\mu_n/\de^{\ve\,2})$ are operators uniformly bounded in $\ve$ for all $n=0,1,2,\dots$; we shall prove the boundedness only of $G_1^\ve(z+\mu_n/\de^{\ve\,2})$, the proof for 
$G_2\lf(z+\frac{\mu_n}{\de^{\ve\,2}}\ri)$ is similar and will be omitted.

We denote by $\mathcal P^\ve(d\la)$ the projector valued measure associated to the operator $\displaystyle-\frac{d^2}{ds^2}-\frac{1}{\ve^2}\frac{\gamma^2(s/\ve)}{4}$ and by $\sigma^\ve$ its spectrum. Given a function $q\in L^2(\RE)$,  by the spectral theorem we have
\begin{equation*}
\Big\|G^\ve_{1}\lf(z+\frac{\mu_n}{\de^{\ve\,2}}\ri)q\Big\|_{L^2}^2 =
\int_{\sigma^\ve}
\sum_{j=0}^\infty 
\frac{\xi_j^2(d)}{|\la+(\mu_j-\mu_n)/\de^{\ve\,2}-z|^2}
(q,\mathcal P^\ve(d\la)q)_{L^2(\RE)}\,.
\end{equation*}
It is easy to see that for $\Im z \neq 0$ 
\begin{equation*}
\sum_{j=0}^\infty 
\frac{\xi_j^2(d)}{|\la+(\mu_j-\mu_n)/\de^{\ve\,2}-z|^2}\leqslant c\,,
\end{equation*}
this is proved by decomposing the series into two parts
\begin{equation*}
 \sum_{j=0}^\infty\frac{\xi_j^2(d)}{|\la+(\mu_j-\mu_n)/\de^{\ve\,2}-z|^2}
\leqslant
\de^{\ve\,4}\sum_{j\leqslant n}\frac{\xi_j^2(d)}{|\mu_j-\mu_n+\de^{\ve\,2}(\la-z)|^2}+
\de^{\ve\,4}\sum_{j>n}\frac{\xi_j^2(d)}{|\mu_j-\mu_n+\de^{\ve\,2}(\la-z)|^2}\,.
\end{equation*}
The first term in the right hand side can be easily estimated by the imaginary part of $z$
\begin{equation*}
\de^{\ve\,4}\sum_{j\leqslant n}
\frac{\xi_j^2(d)}{|\mu_j-\mu_n+\de^{\ve\,2}(\la-z)|^2}\leqslant\frac{c}{|\Im z|}\,.
\end{equation*}
To estimate the last term we notice that, since $\la\in\sigma^\ve\subseteq[-c/\ve^2,\infty)$ there exists $\ve_0$ such that for all $0<\ve<\ve_0$  and for all $j>n$
\begin{equation*}
|\mu_j-\mu_n+\de^{\ve\,2}(\la-z)|\geqslant
|\mu_j-\mu_n-c(\de^{\ve}/\ve)^2|\,
\end{equation*}
from which it follows that for all $0<\ve<\ve_0$ and uniformly  in $\la\in\sigma^\ve$
\begin{equation*}
\sum_{j>n}\frac{\xi_j^2(d)}{|\mu_j-\mu_n+\de^{\ve\,2}(\la-z)|^2}\leqslant c\,.
\end{equation*}
Therefore from the definition of spectral projection we get 
\begin{equation*}
\Big\|G^\ve_{1}\lf(z+\frac{\mu_n}{\de^{\ve\,2}}\ri)q\Big\|_{L^2} \leqslant
c \|q\|_{L^2(\RE)}\,,
\end{equation*}
and then $G_1\lf(z+\frac{\mu_n}{\de^{\ve\,2}}\ri)$ is bounded uniformly in $\ve$. The boundedness of ${\mathcal G}^{\ve}(z)$ comes from a similar argument and the proof will be omitted. The same estimates obviously hold for ${\mathcal G}^{\ve\,\ast}(z)$.

By noticing that for all $q\in L^2(\RE)$
\begin{equation*}
\bigg\|G^\ve_{1,1}\Big(z+\frac{\mu_n}{\de^{\ve\,2}}\Big)q\bigg\|_{L^2(\RE)}^2 \leqslant 
\int_{\sigma^\ve}\bigg|
\sum_{j=0}^\infty\frac{\xi_j^2(d)}{\la+(\mu_j-\mu_n)/\de^{\ve\,2}-z}\bigg|^2\,
(q,\mathcal P^\ve(\la)q)_{L^2(\RE)}\,d\la\,.
\end{equation*}
and by   an argument similar to the one used before one can show that for all $\Im z\neq0$ there exists an $\ve_0$ such that for all $0<\ve<\ve_0$
\begin{equation*}
\bigg|
\sum_{j=0}^\infty\frac{\xi_j^2(d)}{\la+(\mu_j-\mu_n)/\de^{\ve\,2}-z}\bigg|^2\leqslant c\,.
\end{equation*}
From which it follows that $G_{1,1}^\ve(z+\mu_n/\de^{\ve\,2})$ is uniformly bounded in $\ve$. The same holds for all $G_{i,j}^\ve(z+\mu_n/\de^{\ve\,2})$ and, consequently, for $\Gamma^\ve(z+\mu_n/\de^{\ve\,2})$. The boundedness of $\Gamma^\ve(z)$ easily comes. 
 
Now we can go back to formula \eqref{charge} and prove that there exists a unique solution $\underline q^\ve\in L^2(\partial\Omega')$. We have just proved that there exists $\ve_0$ such that, for $0< \ve < \ve_0$, 
$ \| (\underline{\al} - \underline{\al}^\ve ) \Gamma^\ve(z) \|_{L^2(\partial\Omega')} <1$, then
$(I + (\underline{\al} - \underline{\al}^\ve ) \Gamma^\ve(z) )^{-1}$ is well defined by its Neumann series
and
\begin{equation*}
\underline{q}^\ve = - \big( I + (\underline{\al} - \underline{\al}^\ve ) \Gamma^\ve(z) \big)^{-1} 
(\underline{\al} - \underline{\al}^\ve )  {\mathcal G}^{\ve\,\ast}(z) \psi 
\end{equation*}
is the solution of \eqref{charge}. 

This concludes the proof of formula \ref{risolvente}. From the inequality $\| (\underline{\al} - \underline{\al}^\ve )\|_{L^\infty(\partial\Omega')}\leqslant\de^\ve/\ve$ we also get $\|\underline q^\ve\|_{L^2(\partial\Omega)}\leqslant c \de^\ve/\ve\|\psi\|_{L^2}$ and  we proved all the estimates that we used in lemma \ref{lemma1}.

\end{document}